\documentclass{ifcolog-c}

\title{Substitution in the $\lambda$-Calculus and the role of the Curry School
}
\titlerunning{Substitution in the $\lambda$-Calculus}
\titlethanks{This paper appears in  {\em A Century since Principia's Substitution Bedazzled Haskell Curry}, Fairouz Kamareddine (ed).  College Publications (11 July 2023). ISSN-10 1848904363, ISBN 13 978-1848904361. Here is the list of contents: \url{http://collegepublications.co.uk/contents/tbt00051.pdf}}
\addauthor[]{Fairouz Kamareddine}{}
\authorrunning{Kamareddine}

\usepackage{amsmath,amsthm,amssymb}
\usepackage{color}
\usepackage{tcolorbox}

\usepackage{flag}
\usepackage{ifthen}

\newcounter{flagcounter}
\newcounter{flagline}

\newlength{\derivSkip}
                              
\newlength{\myheight}
\newcommand{\makehigher}[2]{
   \settoheight{\myheight}{$#2$}
   \addtolength{\myheight}{#1}
   \raisebox{0pt}[\myheight]{$#2$}
   }

\newcommand{\flagbox}[1]{
         \setlength{\fboxsep}{0.5ex}
         \fbox{$#1$}}

\newcommand{\derline}[3]
{\refstepcounter{flagline}
 & 
 \ifthenelse
     {\value{flagcounter}=0}
     {\begin{array}[b]{l}
         \ \makehigher{\derivSkip}{#2}
      \end{array}
     } 
     {\begin{array}[b]{*{\value{flagcounter}}{@{\hspace{2\arraycolsep}}|}l}
         \ \makehigher{\derivSkip}{#2}
      \end{array}
     }
 & 
 \mbox{#3} 
 \\
}

\newcommand{\flagstart}[3]
{\refstepcounter{flagline}
 & 
 \ifthenelse
     {\value{flagcounter}=0}
     {\begin{array}[b]{@{\hspace{2\arraycolsep}}l} 
         \makehigher{\derivSkip}{\flagbox{#2}}
      \end{array}
     }
     {\begin{array}[b]{*{\value{flagcounter}}{@{\hspace{2\arraycolsep}}|}@{\hspace{2\arraycolsep}}l}
         \makehigher{\derivSkip}{\flagbox{#2}}
      \end{array}
     }
 \addtocounter{flagcounter}{1}
 & 
 \mbox{#3}
 \\
}

\usepackage{pifont}
\newcommand{\xmark}{\ding{55}}%
\newcommand*\circled[1]{\tikz[baseline=(char.base)]{
    \node[shape=circle,draw,inner sep=1pt] (char) {#1};}}

\usepackage{graphicx}
\usepackage{amssymb} 

\usepackage{jbw-debug}
\usepackage{jbw-math}
\usepackage{jbw-arrow}
\usepackage{jbw-mn-symbol}

\usepackage{textgreek}

\newcounter{parens}
\def\countlparen{%
    \addtocounter{parens}{1}\lparen_{\ensuremath{_{\the\value{parens}}}}%
}
\def\countrparen{%
    \rparen_{\ensuremath{_{\the\value{parens}}}}\addtocounter{parens}{-1}%
}
\let\lparen(
\let\rparen)
\begingroup
    \catcode`(\active
    \catcode`)\active
    \gdef\countparens{%
        \let(\countlparen
        \let)\countrparen
    }
\endgroup   

\newenvironment{nested parentheses}
{%
    \catcode`(\active
    \catcode`)\active
    \countparens
    \setcounter{parens}{0}%
}
{}

\def\disp{\displaystyle}

\newcommand{\mytfunsym}{\stackrel{\hbox{\tiny tot}}{\longrightarrow}}

\newcommand\mypfun{\stackrel{\circ}{\longrightarrow}}

\newcommand{\la}{\lambda}
\newcommand{\La}{\Lambda}
\newcommand \fik {\varphi^i_k}
\newcommand \fikl {\varphi^i_{k+1}}
\newcommand \lar {\longrightarrow}
\newcommand \ls {\lambda\sigma}

\newcommand \uik {U^i_k}
\newcommand \uikl {U^i_{k+1}}
\newcommand \las {\la s}
\newcommand \Las {\La s}
\newcommand \si {\,\sigma^i}
\newcommand \sj {\,\sigma^j}

\newcommand \sil {\,\sigma^{i+1}\,}
\newcommand{\tea}{\ate}
\newcommand{\teb}{\bte}
\newcommand{\tec}{\cte}

\newcommand \um {{\texttt 1}}
\newcommand{\MetaSubst}[3]{{#1}\{\!\!\{{\texttt #2}\!\leftarrow \!{#3}\}\!\!\}}
\newcommand{\all}{\MetaSubst}

\newcommand{\be}{\beta}
\newcommand{\beal}{\be'}
\newcommand{\beall}{\be''}
\newcommand{\cM}{\cal M}
\newcommand{\mycC}{\cal C}
\newcommand{\ate}{A}
\newcommand{\bte}{B}
\newcommand{\cte}{C}
\newcommand{\den}[1]{\llbracket #1 \rrbracket}

\newcommand{\cut}[1]{}

\def\rrightarrow{\rightarrow \hspace{-.8em} \rightarrow}
\newcommand{\alre}{\mbox{$\rightarrow_{\alpha}$}}

\newtheorem{lemma}{Lemma}

\newtheorem{example}{Example}

 \newtheorem{definition}{Definition}
 \newtheorem{convention}[definition]{Convention}
 \usepackage{epsfig, tikz}
\usetikzlibrary{matrix, graphs}
 
\begin{document}
 \nopagenumber
\maketitle
\begin{abstract}
Substitution plays a prominent role in the foundation and implementation of mathematics and computation.   
In the $\la$-calculus, we cannot define $\alpha$-congruence without a form of substitution but for substitution and reduction to work, we need to assume a form of $\alpha$-congruence (e.g., when we take $\la$-terms modulo bound variables). Students on a $\la$-calculus course usually find this confusing.    The elegant writings and research of the Curry school~\cite{curry/hindley/seldin:CLgcII,hindley/lercher/seldin:ICLg,hindley/seldin:HBCE} have settled this problem very well.  This article is an ode to the contributions of the Curry school (especially the excellent book of Hindley and Seldin~\cite{hindley/seldin:ICLC,hindley/seldin:LCCI})  on the subject of $\alpha$-congruence  and substitution.

\end{abstract}

\section{Introduction}
\label{introsec}
Jonathan Seldin completed his PhD thesis ({\em Studies in Illative Combinatory Logic}~\cite{seldin:SICL}) in 1968 under the supervision of Haskell Curry at the University of Amsterdam, and to this day, he has safeguarded all the manuscripts of Curry as well as his legacy.  This article is dedicated to Jonathan for the great role he has played throughout the evolution of combinatory logic and the $\la$-calculus. 

The first half of the twentieth century  had seen an explosion of new novel ideas that would shape the foundations of mathematics and would lead to the birth of computability.  During those impressive times, new formal languages, logics, and tools  were invented that up to this day still represent the undisputed standards for formalisation and computation.  These standards include   the Turing machine, the $\lambda$-calculus, category theory, combinatory logic and G\"odel's incompleteness.  Although invented by different people, they are all interconnected and each sheds light on  the {\em Entscheidungsproblem} in its own way.  Alonzo Church gave us the $\lambda$-calculus which is the language of the computable ($f$ is computable iff $f$ is $\lambda$-definable), Haskell Curry gave us Combinatory Logic which is another computation model that is equivalent to the $\lambda$-claculus.  Curry also gave us a number of
concepts that we continue to heavily use today (e.g., the  Curry-Howard isomorphism which gives a correspondence between proofs in proof systems and terms in models of computation,  and the Currying notion where a function of many arguments can be evaluated as a sequence of unary-functions, as well as  the \emph{functionality} concept, which became the basis of what nowadays is called type assignment).  Turing gave us another equivalent model of computation: the Turing machine
which is the machine of the computable ($f : \mathbb{N} \mapsto \mathbb{N}$ is computable iff there is a Turing machine $M_f$ that takes an input $n$ and halts with output $f(n)$).\footnote{Samuel Eilenberg and Saunders Mac Lane gave us category theory which is an elegant foundation of mathematics. And there was the new developments of set theory and type theory which appropriately claim their crucial role in the foundations of mathematics and computation.}    In this paper, we only focus on the $\la$-calculus and on some of the lessons from the Curry school on subtitution.

Substitution plays a prominent role in the foundation and implementation of mathematics and computation.
  Bertrand Russell and others made a number of attempts at defining substitution but most, if not all, of the attempts to correctly define substitution with bound variables before the publication in 1934 of Hilbert and Bernays foundations of mathematics~\cite{hilbert/bernays:GLMt.I} were erroneous.  
Although Russell's  type free substitutional theory of 1905~\cite{russelldenoting} enabled him to prove the axiom of infinity, it  also allowed contradictions and  he was frustrated by the paradoxes which as he explained in a letter to his friend the mathematician Ralph Hawtrey "pilled the substitution theory" (see  \cite{urquhart}).  The notion of substitution remained unsatisfactorily explained in his Principia Mathematica.

This problem of substitution was the motivation of Curry\footnote{For Curry's PhD thesis (which was originally written in German) translated into English, see~\cite{kamareddine/seldin2}.}  to develop combinatory logic. In 1922, as a graduate student, Curry noticed in  the first chapter of Principia  \cite{whitehead/russell:PMth} that the rule of substitution of well-formed formulas for propositional variables (which does not even involve the well known complication of bound variables) was considerably more tricky than the rule of detachment (which is equivalent to modus ponens).   The complication Curry noticed in the rule of substitution in Chapter 1 of Principia is now considered to be the complication of its implementation by a computer program although there were no electronic computers when Curry noticed this.\footnote{It is precisely when  Curry started looking for simpler forms of substitution~\cite{curry29, curry52}, that he introduced some of the combinators and  defined in~\cite{curry:AVSC},  what is nowadays called the \emph{bracket abstract} 
$
[x_1, x_2, \cdots , x_n]{\frak{X}}.  
$
Curry was able
to prove that for any term $\frak{X}$ where free variables
 $x_1, x_2, \cdots , x_n$
may appear, there is a term $X$  in which none of these
variables appear with the property that $Xx_1x_2\cdots x_n = \frak{X} .$  This term $X$ is unique by extensionality and is denoted by $[x_1, x_2, \cdots , x_n]{\frak{X}}$.
}

This article discusses  substitution in the $\lambda$-calculus and the role of the Curry school.
One of the first hurdles after introducing students to the extremely simple syntax of the type-free pure $\lambda$-calculus (variables, abstraction and application) is how/when to introduce the substitution rule and  $\alpha$-congruence. In particular, should one use $\alpha$-congruence (yes they should because  $\alpha$-congruence is needed for the Church-Rosser property of $\beta$-reduction\footnote{See Examples A1.1 and A1.2 of \cite{hindley/seldin:LCCI}.}), and then, in which order should one introduce $\alpha$-conversion  and substitution? We cannot define $\alpha$-congruence without a form of substitution but for substitution and reduction to work, we need to assume a form of $\alpha$-congruence (e.g., when we take $\la$-terms modulo bound variables). Students on a $\la$-calculus course usually find this confusing. Moreover, to them, this does not fit with  a purely syntactical account where much of the work should be carried out using rules that can be automated.  This was behind most research on expliciting the notion of substitution in the $\lambda$-calculus in order to bridge theory  and implementation.

The more the students appreciate  the enormous intellectual debates triggered by the paradoxes and the labour that led to the birth of computers and the  foundations and limits of computation,\footnote{Ranging from    Frege's abstraction principle, his general concept of a function and his formalisations~\cite{begrif,grund}, to Russell's paradox in Frege's work and his monumental Principia Mathematica \cite{whitehead/russell:PMth}, to Hilbert's   Entscheidungsproblem and the languages/models of Curry, Turing, Church and others.} the more they question representations that might confuse them.  
        Students prefer rules that they can apply or even implement to being told to work things in their head.
They get confusd if they are told at the start  that $\la$-terms will be taken modulo the name of their bound variables  and then shown how to build $\alpha$-reduction using substitution and to give substitution modulo $\alpha$-equivalence classes. They ask which is first ($\alpha$ or substitution). For them this is  bothersome because they feel that  the dependence between $\alpha$-reduction and substitution is self-referential.

The elegant writings and research of the Curry school~\cite{curry/hindley/seldin:CLgcII,hindley/lercher/seldin:ICLg,hindley/seldin:HBCE} have settled this problem very well.  This article is an ode to the contributions of the Curry school (especially the excellent book of Hindley and Seldin~\cite{hindley/seldin:ICLC,hindley/seldin:LCCI})  on the subject of $\alpha$-reduction
and substitution.

\medskip
In Section~\ref{secsynsemnot} we introduce the basic syntax (with the purely syntactic identity $=_{\cal M}$), notational conventions and denotational semantics of the $\la$-calculus.  In Section~\ref{compsec} we move to computation and calculation 
with $\la$-terms  and explain the problem of 
the grafting $A\{v:=B\}$.\footnote{Which is the purely syntactic replacement/substitution of a variable $v$ inside a term $A$  by a term $B$.} Then we introduce bound and free occurrences of variables as well as the replacement  $A\langle\langle v:=B\rangle\rangle$ based on an ordered variable list.\footnote{Which replaces/substitutes a variable $v$ inside a term $A$  by a term $B$ using an ordered variable list to avoid the capture of free variables.} We note that although calculating (which we called $\overline{\be}$-reducing) $\la$-terms using the substitution based on an ordered variable list avoids the problem of variable capture, it does not work on its own because
the order in which a term is computed will affect the answer.  In Section~\ref{syntidenrevi} we  discuss two alternatives to identify terms modulo bound variables.  One approach ($=_{\alpha}$) builds $\alpha$-congruence directly from  the replacement based on ordered variables while the other 
is based on safe applications of grafting to build a notion $=_{\alpha'}$ of term equivalence.
Interestingly, both approaches give the same notion of syntactic equivalence denoted by $\equiv$ where terms that only differ in the name of their bound variables are equivalent.  In Section~\ref{secbetaalpha} we introduce the $\beal$ and $\beall$ reduction relations.
The basic ($\beal$) rule is the same as (${\overline{\beta}}$) but  the reflexive transitive closure $\rrightarrow_{\beal}$ adds $\alpha$-reduction in the sense  that if $A \rrightarrow_{\beal} B$ then there is $C$  such that $A\rrightarrow_{\overline{\beta}} C\rrightarrow_{\alpha} B$.  The basic ($\beall$) adds $\alpha'$-reduction to transform the redex into a so-called clean term before any beta reduction takes place and then safely uses grafting since this does not cause any problems with clean terms. Not only the basic  ($\beall$) rule is based on $\alpha'$-reduction, but also the reflexive transitive closure $\rrightarrow_{\beall}$.  So, for grafting to work, we use   $\alpha'$-reduction at the basic reduction stage (to clean the term) and at the reflexive transitive stage.
Unlike $\overline{\be}$-reduction, both $\beal$ and $\beall$ reductions now satisfy CR.  However, there is an oddity about $\beall$ compared to $\beal$.  ($\beal$) is a function in the sense that, if $(\la v.A)B\rightarrow_{\beal} C$ then $C$ is unique because the replacement relation for $\beal$ is based on an ordered variable list.\footnote{Note that although ($\beal$) is a function, $\rightarrow_{\beal}$ is not unless we take it as a function of 2 arguments, the term and the particular redex-occurrence.  For example, $(\la x.\underline{(\la y.y)x})u \rightarrow_{\beal} (\la x.x)u$ and  $\underline{(\la x.(\la y.y)x)u} \rightarrow_{\beal} (\la y.y)u$ but  $(\la x.x)u\not=_{\cal M} (\la y.y)u$.}     On the other hand, ($\beall$) is not a function  since
($\beall$) uses $\alpha'$-reduction to clean terms.   For example,  we can clean
  $(\la x.x(\la x.x))x$ into either $(\la z.z(\la y.y))x$ or $(\la y.y(\la z.z))x$ and so we get 
$(\la x.x(\la x.x))x  \rightarrow_{\beall} x(\la y.y)$ and $(\la x.x(\la x.x))x  \rightarrow_{\beall} x(\la z.z)$ yet  $x(\la y.y)\not=_{\cal M} x(\la z.z)$. But since  both $\beal$ and $\beall$ reductions have used $\alpha/\alpha'$ reduction and since we have built and understood the priorities and order of substitution and equivalence classes modulo the names of bound variables ($\alpha/\alpha'$), we can now move to the stage at which courses on the $\la$-calculus usually start and take terms modulo the names of bound variables.   So in the rest of Section~\ref{secbetaalpha}, we introduce the usual notion of substitution and $\be$-reduction and give a summary and comparison of all the $\be$-reductions introduced so far.  For the sake of completeness, in Section~\ref{secdeb} we present the $\la$-calculus with de Bruijn indices where the de Bruijn indices incorporate the equivalence classes of terms.   We look at the substitution and reduction rules with de Bruijn indices which will lead us naturally to a calculus with explicit substitutions where we take the classical $\la$-calculus with de Bruijn indices exactly as it is, but simply turn its meta-updating and meta-substitution to the object level to obtain a calculus of explicit substitutions. Extending $\la$-calculi with explicit substitutions is essential for the implementations of these calculi.  
We conclude in Section~\ref{secconc}.

\section{The Syntax and Denotational Semantics of the $\la$-calculus}
\label{secsynsemnot}
\begin{convention}
  \label{convsymb}
  If $\mbox{Symb}$ ranges over a set of entities ${\cal A}$ then also $\mbox{Symb}'$,  $\mbox{Symb}''$,  $\mbox{Symb}_1$,  $\mbox{Symb}_2$, etc., 
  range over ${\cal A}$.
  \end{convention}
\begin{definition}[$\lambda$-terms in ${\cal M}$]
  \label{lambdatermsdefs}
   \begin{itemize}
\item Let ${\cal V} = \{x,y,z, x',y',z', x_1,y_1,\break z_1, \dots\}$ be
  an infinite set of \em{term variables} and
    let \em{meta variables} $u, v$ range over $\cal V$.
    By  Convention~\ref{convsymb}, also $u'$, $v_1$, etc., range over $\cal V$.\\
    Note that   ${\cal V}$ is a set and hence all its elements are pairwise distinct.
    \item
The set of classical $\lambda$-terms (or $\la$-expressions) $\cM$ is given by:
 \begin{center}
$\cM  \: ::=  \:  {\cal V} \:|\: ( \la{{\cal V}}.{\cM}) \:|\: ( \cM \cM)$.
 \end{center}
 We let capital letters $ \ate,\bte,\cte, D,  \cdots$ range over  $\cM$.
    \item
The set of contexts with one hole $\mycC$ is given by:
 \begin{center}
$\mycC  \: ::=  \:  [\:] \:|\: ( \la{{\cal V}}.{\mycC}) \:|\: ( \mycC \cM)\:|\: ( \cM \mycC)$.
 \end{center}
 We call $[\:]$ a hole and  let  $ C[\:],  C'[\:],  C_1[\:],   \cdots$ range over  $\mycC$.\\
 We write $C[A]$ to denote the filling of the hole in $C[\:]$ with 
 term $A$.\\
 For example, if $C[\:]$ is $(\la x.([\:]x))$ then  $C[x]$ is $(\la x.(xx))$,
 $C[y]$ is $(\la x.(yx))$, and $C[(\la y.y)]$ is $(\la x.((\la y.y)x))$. 
\item
  The {\em length} of a term $A$ (written $\#A$) is defined inductively as follows:\\
  $\#v = 1\hspace{0.4in} $ $\#(AB) = \#A + \#B\hspace{0.4in}$ $\#(\la v.A) = 1+\#A$.
 \end{itemize}
\end{definition}

\begin{definition}[Compatibility]
\label{defcompatibility}
We say that a relation\footnote{Note that if $R$ is a compatible relation and $A R B$, then $C[A]RC[B]$ for any context $C[\:]$.} $R$ on $\la$-terms is {\em compatible} if the following hold for any $\la$-terms  $A, B$ and $C$ and for any variable $v$:
$$\disp\frac{A R B}{(AC) R (BC)} \hspace{0.5in}
\disp\frac{A R B}{(CA) R (CB)}  \hspace{0.5in}
\disp\frac{A R B}{(\la v.A) R (\la v.B)}$$
\end{definition}
Here,  $\disp\frac{\mbox{above}}{\mbox{below}}$
means that  if ``above'' holds then ``below'' holds too.

\begin{definition}[Strict Equality $=_{\cal M}$]
  \label{defequalitys1}
  Strict equaliy $=_{\cal M}$ on the set ${\cal M}$ of $\la$-terms is defined as  the reflexive, transitive, symmetric and compatible closure under
  $$\disp \frac{\mbox{$v \in {\cal V}$}}{v =_{}v}\hspace{0.1in}\hfill{(\mbox{Base $=_{\cal M}$})}$$
  Note that Strict Equality $=_{\cal M}$ is $\{(A,A) | A\in {\cal M}\}$.
\end{definition}

To explain the meaning of terms, let us imagine
a model ${\cal D}$\index{${\cal D}$}  where every $\lambda$-term denotes an element of that
model.  We let $\textbf{d}$ range over  ${\cal D}$ (and hence by  Convention~\ref{convsymb}, $\textbf{d}'$, $\textbf{d}''$, $\textbf{d}_1$, $\textbf{d}_2$ etc., also range over  ${\cal D}$).

The meaning of terms depends very much on the values we assign to variables.  This is why we need the so-called  {\em environment}.

\begin{definition}[$\texttt{ENV}$, Changing Environments]
We define the set of {\em environments}\index{environment}  $\texttt{ENV}$\index{$\texttt{ENV}$} as the set of total functions from  ${\cal V}$ to  ${\cal D}$.  We let $\sigma$ range over $\texttt{ENV}$.

$$\texttt{ENV} =  {\cal V}  \mytfunsym {\cal D}.$$

We define the new environment $\sigma[\textbf{d}/v]\in {\cal V} \mytfunsym {\cal D}$ as follows:
$$\sigma[\textbf{d}/v](v')=\begin{cases}
\textbf{d} & \mbox{ if $v' = v$}\\
\sigma(v') & \mbox{otherwise}
\end{cases}
$$
\end{definition}
Note that 
      $\sigma[\textbf{d}
           /v][\textbf{d}'
            /v](v) = \textbf{d}'
      $, 
          $\sigma[\textbf{d}
      /v][\textbf{d}'
      /v] = \sigma[\textbf{d}'
      /v]$ and 
             if $v\not = v'$ then $\sigma[\textbf{d}
                /v][\textbf{d}'
                /v']= \sigma[\textbf{d}'
                /v'][\textbf{d}
                /v]$.
 
\begin{definition}[Denotational meaning of terms]\index{meaning ! term}
  \label{meaningfuncdef}\hfill\mbox{}
  We define the following:\footnote{Note that we defined our environments to be total functions on ${\cal V}$ ( $\texttt{ENV} =  {\cal V}  \mytfunsym {\cal D}$) and we defined our meaning function $\llbracket \: \rrbracket \in {\cal M}\mytfunsym(\texttt{ENV} \mytfunsym {\cal D}) $
to  be a total function on ${\cal M}$  and hence, for any $\sigma\in \texttt{ENV}$ and $A\in{\cal M}$, 
$\den{A}_\sigma$ is defined and is an element of ${\cal D}$.  We also see that when $A$ is an abstraction $(\la v.B)$, then  $\den{(\la v.B)}_\sigma$ is a partial function in ${\cal D} \mypfun {\cal D}$. }
  \begin{itemize}
    \item
      The function  $\llbracket \: \rrbracket \in {\cal M}\mytfunsym(\texttt{ENV} \mytfunsym {\cal D}) $
 as follows:\index{$\den{}$}
 \begin{itemize}
   \item
$\den{v}_\sigma = \sigma(v)$.
   \item
$\den{(AB)}_\sigma = \den{A}_\sigma(\den{B}_\sigma)$.
   \item
     $\den{(\la v. A)}_\sigma = \textbf{f}\in {\cal D} \mypfun {\cal D}$
     such that $\forall \textbf{d} \in {\cal D}$,
     $\textbf{f}(\textbf{d}) = \den{A}_{\sigma[\textbf{d}/v]}$.
     \end{itemize}
\item
  $A$ and $B$ {\em have the same meaning in environment $\sigma$} iff $\den{A}_\sigma = \den{B}_\sigma$.
  \item
    $\den{A} = \den{B}$
iff $\den{A}_\sigma = \den{B}_\sigma$  for every   environment $\sigma$.
\end{itemize}
\end{definition}

We
 will assume the usual parenthesis convention where 
application associates to the left, 
outer parenthesis 
can be dropped, 
the body of a $\lambda$ includes everything that comes after it and a 
sequence of $\la$'s is compressed to one.  So, 
 $ABCD$ denotes $(((AB)C)D)$ and 
  $\la vv'.AB$ denotes
$\la v.(\la v'.(AB))$.

 \section{Computing $\la$-terms}
  \label{compsec}
 \begin{center}
   \begin{tcolorbox}[width=.9\textwidth, colframe=blue]
     The abstraction term $(\lambda v. A)$ alone is a bachelor waiting for its partner.  As soon as a $B$ stands to the right of the abstraction $(\lambda v. A)$ we get a partnered couple  $(\lambda v. A)B$.  When the time comes and we are ready for computation, $(\lambda v. A)B$ can produce a child ${\ate}_{v:=\bte}$ which is the body $A$ of the abstraction in which all the occurrences of $v$ are replaced by $B$ as follows:
\begin{equation}
(\la v.\ate)\bte \mbox{ computes to  } {\ate}_{v:=\bte}\tag{1}\label{eq:4.1}
\end{equation}
This replacement ${\ate}_{v:=\bte}$ still needs to be defined.
          \end{tcolorbox}
 \end{center}

 \subsection{Defining ${\ate}_{v:=\bte}$ as Grafting $A\{v:=B\}$ does not work}
\label{graftnogoodsection}
One might immediately think of defining  ${\ate}_{v:=\bte}$ by a  strict syntactic replacement (a.k.a.\ grafting) as given in Definition~\ref{graftingdef} below, however, this is the wrong definition as we will see in Table~\ref{graftingnogood}.

\begin{definition}[Grafting $A\{v:=B\}$]\index{grafting}\index{$A\{v:=B\}$}
\label{newgraftingdef}
\label{graftingdef}
For any $A, B, v$, we define the {\em grafting} relation $A\{v:=B\}$\index{$A\{v:=B\}$} to be the result of syntactically replacing $B$ for every occurrence of $v$ in $A$, as follows:\footnote{We use $=_{def}$ here instead of $_{\cal M}$ to draw attention that this definition will not work as we will see in the rest of this paper. }
\[
\begin{array}{llll}
1. \: v\{v:=B\} & =_{def}& B  \\
2. \:  v'\{v:=B\} & =_{def} & v'  \hspace{0.5in} \mbox{if $v \not =_{\cal M} v'$}\\
3. \: (A C)\{v:=B\} & =_{def} & A\{v:=B\} C\{v:=B\}  \\
4. \: (\la{v}.{A})\{v:=B\} & =_{def} & \la{v}.{A}  \\
5. \: (\la{v'}.{A})\{v:=B\} & =_{def} & \la{v'}.{A\{v:=B\}}    \hspace{0.5in} \mbox{if $v \not =_{\cal M} v'$}
\end{array}
\]
\end{definition}

\begin{center}
\begin{table}[h]
\begin{tabular}{|l|}
  \hline
  Let $\sigma$ be an environment.\\ Let  $\textbf{f}, \textbf{g},  \textbf{g}'$ and $\textbf{h}$ be such that 
$\forall \textbf{d}, \textbf{d}' \in {\cal D}$:
  $\textbf{f}(\textbf{d})(\textbf{d}') = \textbf{d}(\textbf{d}')$, \\$\textbf{g}(\textbf{d}) = \den{y}_\sigma(\textbf{d})$, $\textbf{g}'(\textbf{d}) = \den{x}_\sigma(\textbf{d})$ and $\textbf{h}(\textbf{d}) = \textbf{d}(\textbf{d})$. \\ Clearly  $\textbf{g}$, $\textbf{g}'$ and $\textbf{h}$ are different and  $\den{\la z.xz}_\sigma = \den{\la y.xy}_\sigma = \textbf{g}'$.  \\
  If we allow Definition~\ref{meaningfuncdef} to cover the $=_{def}$ of Grafting Definition~\ref{graftingdef}\\ (i.e., if $A\{v:=B\} =_{def} C$ then $\den{A\{v:=B\}}_\sigma = \den{C}_\sigma$), 
  then:\\
      $\bullet$     $\den{(\la z.xz)\{x:=y\}}_\sigma = \den{\la z.yz}_\sigma =\textbf{g}$\\
     $\bullet$   $\den{(\la y.xy)\{x:=y\}}_\sigma = \den{\la y.yy}_\sigma =\textbf{h}$.\\
 $\bullet$  
  So,  $\den{\la z.xz}_\sigma = \den{\la y.xy}_\sigma$ but $\den{(\la z.xz)\{x:=y\}}_\sigma \not = $\\ $\qquad \den{(\la y.xy)\{x:=y\}}_\sigma$.\\
 $\phantom \quad\hspace{0.5in}$ This is not good.\\
 $\bullet$ If  $({\la  v.A})_{v:=B}$ is the grafting $A\{v:=B\}$ of  Definition~\ref{graftingdef},\\
 \qquad  then 
  by \eqref{eq:4.1}:\\
 $\phantom\quad\phantom\quad$ 1. $(\la  xz.xz)y  \mbox{  computes to  } (\la  z.xz)\{x:=y\}=_{def}\la z.yz$.  \\ $\phantom \quad\hspace{0.5in}$ And, $\den{(\la xz.xz)y}_\sigma = \textbf{g}= \den{\la z.yz}_\sigma$.\\
  $\phantom\quad\phantom\quad$  2. $(\la  xy.xy)y  \mbox{  computes to  } (\la  y.xy)\{x:=y\}=_{def}\la y.yy$.\\  $\phantom \quad\hspace{0.5in}$ But, $\den{(\la xy.xy)y}_\sigma = \textbf{g}\not= \den{\la y.yy}_\sigma = \textbf{h}$.\\
  This is bad and so we cannot use the computation rule $({\be^w})$:\\
  $\quad\hspace{0.3in}\quad (\la v.A)B \rightarrow_{\be^w} A\{v:=B\}$\\
  because $(\la  xy.xy)y$ and $(\la  xz.xz)y$ have the same meaning $\textbf{g}$,\\
  $(\la  xy.xy)y\rightarrow_{\be^w} \la y.yy$ and $(\la  xz.xz)y\rightarrow_{\be^w}\la z.yz$\\
  but $\la z.yz$ and $\la y.yy$ have different meanings ($\textbf{g}$ resp.\ $\textbf{h}$).\\
\hline
\end{tabular}
\caption{Grafting does not work}
\label{graftingnogood}
\end{table}
\end{center}

Table~\ref{graftingnogood} gives examples where equation \eqref{eq:4.1} is used with the grafting of Definition~\ref{graftingdef} resulting in comptation rule $({\be^w})$.
As we  see,  computing $(\la  xz.xz)\circled{$y$}$  to $(\la  z.xz)\{x:=\circled{$y$}\}=_{def}\la  z.\circled{$y$}z$ is correct whereas computing
$(\la  xy.xy)\circled{$y$}$ to $(\la  y.xy)\{x:=\circled{$y$}\}=_{def}\la  y.\circled{$y$}y$ is wrong.   Since $\la$ is a binder (like $\forall$), the $\circled{$y$}$ which is free in the original term $(\la  xy.xy)\circled{$y$}$ is now bound in its computation $\la  y.\circled{$y$}y$. So $(\la  y.xy)\{x:=\circled{$y$}\}$ gives the wrong answer and we need to find a different definition of replacement which 
respects the free status of $\circled{$y$}$.

First, we define  the  notions of {\em free} and {\em bound} occurrences of variables.

\subsection{Free and bound occurrences of variables}
\label{freesecs}

\begin{definition}[Term Occurrence]\index{occurrence ! term}
  \label{repeatoccurrencedef}
  \label{occurrencedef}
  We define an {\em occurrence relation}\index{occurrence ! term}\index{subterm} between $\la$-terms as follows:
  \begin{itemize}
  \item
    $A$ occurs in $A$.
  \item
    If $A$ occurs in either $B$ or $C$ then $A$ occurs in $BC$.
  \item
    If $A$ occurs in $B$ or $A =_{\cal M} v$ then $A$ occurs in $\la v.B$.
   \end{itemize}
\end{definition}
We can number the occurrences as in 
$\overline{xy}^{\circ_1}(\lambda x.y(\lambda y.z)(\overline{xy}^{\circ_2}))$.
  When the term $A$ is a variable, we leave out the overline as in $x^{\circ_1}(\lambda x^{\circ_2}.y(\lambda y.z)\break (x^{\circ_3}z))$.

  \begin{definition}[Scope, Free/Bound Occurrences, Combinator]    \label{scopeocc}
    \label{repeatscopeocc}
    \hfill \mbox{}
  \begin{enumerate}
       \item
    For a particular occurrence of a $\la{v}.A$ in a term $C$, we call the occurrence of $A$ {\em the scope} of the occurrence of the $\la v$.
  \item
    We call an occurrence of a variable $v$ in a term $C$,
     \begin{enumerate}
     \item
       {\em bound} if it is
in the scope  of a $\la v$ in $C$.\footnote{That is, if the occurrence of $v$ is 
     inside the term $A$ of a  $\la v.A$ which occurs in  $C$.}
      \item
        {\em  bound and binding}, if it is the $v$ of a $\la v$  in  $C$.
      \item
        {\em free} if it is not bound.
     \end{enumerate}
   \item
     We say that $v$ is bound (resp.\ free) in $C$ if $v$ has at least one binding (resp.\ free) occurrence in $C$.  We write $BV(C)$ (resp.\ $FV(C)$) for the set of bound (resp.\ free) variables of $C$.
  \item
   A {\em closed expression} is an expression in which all occurrences of variables are bound.  A closed expression is also called a {\em combinator}.\index{combinator}
      \end{enumerate}
\end{definition}

  \subsection{Defining ${\ate}_{v:=\bte}$ as Replacement $A\langle\langle v:=B\rangle\rangle$  Using Ordered Variables}
\label{revisitgraft}
In order to avoid the problem of grafting as we saw in Table~\ref{graftingnogood}, replacement needs to be handled with care due to the distinct roles played 
by bound and free occurrences of variables.
Since ${(\la{y}.{x y})}_{x:=\circled{$y$}}$ must not return $\la{y}.{\circled{$y$} y}$ and we should not change the free $\circled{$y$}$, then we should change  the name of the bound variable $y$ in $\la{y}.{x y}$ to $v\not\in\{x,y\}$ obtaining  $\la v.xv$.  Then,  replacing  the $x$  of  $\la v.xv$ by $\circled{$y$}$ gives $\la v.\circled{$y$}v$ which is fine because   $v\not\in\{x,y\}$.
   Clearly we need to change the following rule 5.\ of  Definition~\ref{graftingdef}:
   $$5.\ (\la{v'}.{A})\{v:=B\}  =_{def}  \la{v'}.{A\{v:=B\}}    \hspace{0.5in} \mbox{if $v \not =_{\cal M} v'$}.$$\
   We could split 5.\ into 2 rules as follows:
   \[
   \begin{array}{lll}
5. \: (\la{v'}.{A})\{ v:=B\} & =_{def} & \la{v'}.{A\{ v:=B\}}    \hspace{0.5in} \mbox{if $v \not =_{\cal M} v'$}\\
 & & \mbox{and ($v' \not \in FV(B)$ or $v \not \in FV(A)$)}\\
6. \: (\la{v'}.{A})\{ v:=B\} & =??? & \la{v''}.{A\{ v':=v''}\}\{ v:=B\}  \hspace{0.3in} \mbox{if $v \not =_{\cal M} v'$} \\
 & & \mbox{and ($v' \in FV(B)$ and  $v \in FV(A)$)} \\
& &  \mbox{and $v'' \not \in FV(AB)$}
\end{array}
\]

We still need to precise the $v''$ of rule 6.  Clearly if we stick to the strict equality of Definition~\ref{defequalitys1}, then $v''$ must be unique since if $\la{v''_1}.yv''_1 =_{\cal M}\la{v''_2}.yv''_2$  then $v''_1 =_{\cal M} v''_2$ follows from the following consequence of  Definition~\ref{defequalitys1}:
\begin{equation}
  (\lambda v.A)  =_{\cal M}  (\lambda v'.A') \mbox{ iff 
    ($v =_{\cal M} v'$ and $A  =_{\cal M}  A'$)}\tag{2}
  \label{eqsynt}
\end{equation}

One way to get a unique result in rule 6.\  would be to order the list of variables ${\cal V}$ and
  then to take $v''$ to be the first variable in the ordered list
  ${\cal V}$ which is different from $v$ and $v'$ and which occurs  after all the free variables of $AB$.  To do this, we add to clause 6  the condition that $v''$ is the first variable in the ordered list ${\cal V}$ which satisfies the conditions of clause 6. That is, we replace  Definition~\ref{graftingdef} by:
  \begin{definition}[Replacement  \mbox{$A \langle \langle v:=B\rangle\rangle$} using ordered variables]
\label{orderedsubstdef}
We define $A{\langle \langle v:=B\rangle\rangle }$ to be the result of replacing $B$ for every free occurrence of $v$ in $A$:
\[
\begin{array}{llll}
1. \: v{\langle\langle v:=B\rangle\rangle} & =_{\cal M} & B  \\
2. \:  v'{\langle\langle v:=B\rangle\rangle} & =_{\cal M} & v'  \hspace{0.5in} \mbox{if $v \not =_{\cal M} v'$}\\
3. \: (A C){\langle\langle v:=B\rangle\rangle} & =_{\cal M} & A{\langle \langle v:=B\rangle\rangle} C{\langle\langle v:=B\rangle\rangle}  \\
4. \: (\la{v}.{A}){\langle\langle v:=B\rangle\rangle} & =_{\cal M} & \la{v}.{A}  \\
5. \: (\la{v'}.{A}){\langle\langle v:=B\rangle\rangle} & =_{\cal M} & \la{v'}.{A{\langle\langle v:=B\rangle\rangle}}    \hspace{0.5in} \mbox{if $v \not =_{\cal M} v'$}\\
 & & \mbox{and ($v' \not \in FV(B)$ or $v \not \in FV(A)$)}\\
6. \: (\la{v'}.{A}){\langle\langle v:=B\rangle\rangle} & =_{\cal M} & \la{v''}.{A{\langle\langle v':=v''\rangle\rangle}}{\langle\langle v:=B\rangle\rangle} \\ &&
\mbox{if $v \not =_{\cal M} v'$} \\
 & & \mbox{and ($v' \in FV(B)$ and  $v \in FV(A)$)} \\
& &  \mbox{and $v''$ is the first variable in the ordered} \\
  & & \mbox{variable list ${\cal V}$ such that $v''\not\in FV(AB)$.}
\end{array}
\]
\end{definition}

  \label{orderedlist}
For example,
  if the ascending order in ${\cal V}$ is
\[x, y, z, x', y', z', x'', y'', z'', \dots\]
then $z$  the first variable in this list which is different from $x$ and $y$ and which is not free  in either $y$ or $xy$.  So, 
$(\la {y}.xy)\langle\langle x:=y\rangle\rangle $ can only be  $(\la {z}.yz)$.

The next lemma plays some initial steps of comparing grafting with replacement using ordered variables.  Note especially item~\ref{orderedmeta-sublema4} which shows that under some strict conditions, grafting is replacement.
\begin{lemma}
  \label{orderedmeta-sublema}
  \hfill \mbox{}
  \begin{enumerate}
  \item
  \label{orderedlemmasubstnofreevars}
  If $v\not\in FV(A)$ then for any $B$, $A\langle\langle v:=B\rangle\rangle =_{\cal M} A$.\footnote{Also, if $v\not\in FV(A)$  then $A\{v:=B\} =_{\cal M} A$.}
\item
  \label{orderedlemmafreevarsofasubst}
  If $v\in FV(A)$ then $FV(A\langle\langle v:=B\rangle\rangle) = (FV(A)\setminus \{v\})\cup FV(B)$.\footnote{This does not hold for $A\{v:=B\}$.  E.g.,
    $(FV(\la x.y)\setminus \{y\})\cup\{x\} = \{x\}\not = FV((\la x.y)\{y:=x\}) = FV(\la x.x)=\emptyset$.  However,  $FV(A\{v:=B\}) \subseteq (FV(A)\setminus \{v\})\cup FV(B)$.}
  \item
    \label{orderedmeta-sublema4}
    If   $v'\not\in FV(vA)$, $v,v'\not\in BV(A)$
    then $A\{v:=v'\}  =_{\cal M} A\langle\langle v:=v'\rangle\rangle$.
  \item
  \label{orderedmeta-sublema3}
  If $v\not = v'$, $v\not\in FV(C)$ and whenever $\la v''.D$ occurs in $A$ then $v''\not\in FV(BC)$ (i.e., no bound variable of $A$ occurs free in $BC$), then $A\langle\langle v:=B\rangle\rangle \langle\langle v':=C\rangle\rangle  =_{\cal M} A\langle\langle v':=C\rangle\rangle \langle\langle v:=B\langle\langle v':=C\rangle\rangle \rangle\rangle $.
  \end{enumerate}
\end{lemma}
\begin{proof}
  1.\ and 2.\ are by induction on the derivation $A\langle\langle v:=B\rangle\rangle$.  3.\ is by induction on the derivation $A\{v:=v'\} =_{def} C$.  4.\ is by induction on $A$.
  \end{proof}
With this lemma, we are starting to see some of the complications of having to take the first relevant variable in the ordered variable list in the definition of $A\langle\langle v:=B\rangle\rangle$.  In fact, without the condition ``whenever $\la v''.D$ occurs in $A$ then $v''\not\in FV(BC)$'', we would not be able to prove Lemma~\ref{orderedmeta-sublema}.\ref{orderedmeta-sublema3}. This defeats the purpose of having an ordered list of variables.

However, there is a more substantial reason as to why an ordered list of variables on its own will not work (see Example~\ref{ecampledoesnotwork} below).  As you recall, we are trying to define the replacement ${\ate}_{v:=\bte}$ in order to define the computation of equation (\ref{eq:4.1}).   If we use the replacement Definition~\ref{orderedsubstdef} to define computation, we get this definition:

 \begin{definition}  \label{overlinebetared}
We define $\rightarrow_{\overline{\beta}}$ as the least compatible relation closed under:
\[ ({\overline{\beta}}) \hspace{0.5in} (\la{v}.A)B \rightarrow_{\overline{\beta}} A\langle\langle v:=B\rangle\rangle \hspace{0.5in} \]
We call this reduction relation $\overline{\beta}$-reduction.
We define $\rrightarrow_{\overline{\beta}}$  as the reflexive transitive closure of $\rightarrow_{\overline{\beta}}$.
\end{definition}

 With this definition  we would lose the so-called Church-Rosser (CR) Property which is defined for a relation $R$ as follows:
 \begin{definition}\label{CRdef}
     We say that a relation $R$ on ${\cal M}$  enjoys the CR property if whenever  $A R B$ and $A R C$ then there is $D$ such that $B R D$ and $C RD$.  
   \end{definition}

$\rrightarrow_{\overline{\be}} $ does not enjoy the CR property.  This can be seen as follows (example is taken from \cite{hindley/seldin:LCCI}):
 \begin{example}
   \label{ecampledoesnotwork}
   $(\la xy.yx)((\la z.x')y)\rrightarrow_{\overline{\be}} \la y.yx'$ and
   $(\la xy.yx)((\la z.x')y)\rrightarrow_{\overline{\be}} \la y'.y'x'$ (assuming the ordered variable list given on page~\pageref{orderedlist}).  It is clear that
   $\la y.yx'\not=_{\cal M}\la y'.y'x'$ and there is no $D$ such that $\la y.yx'\rrightarrow_{\overline{\be}}D$ and $\la y'.y'x'\rrightarrow_{\overline{\be}}D$.  Note that $\den{\la y.yx'}=\den{\la y'.y'x'}$.
        \end{example}

\section{Syntactic identity revised, Searching for $\equiv$}
\label{syntidenrevi}
In the previous section we discussed computation using either 
grafting 
($\be^w$ which does not work) or replacement based on ordered variables ($\overline{\be}$ which has its complications, and moreover, will still not work as we saw in Example~\ref{ecampledoesnotwork}).  

In this section we will discuss two alternatives both of which identify terms modulo bound variables.
One approach builds the so-called $\alpha$-congruence $=_{\alpha}$ directly from  the replacement based on ordered variables given in Definition~\ref{orderedsubstdef}   (see Definition~\ref{alphared}) 
while the other is based on safe applications of the grafting of Definition~\ref{graftingdef} (see Definition~\ref{syntactiequidef}) to build a notion
$=_{\alpha'}$ of term equivalence.
Interestingly, both approaches give the same notion of syntactic equivalence denoted by $\equiv$ where terms that only differ in the name of their bound variables are equivalent.
For example:  $\la{y}.{xy}\equiv\la{z}.{xz}$, $\la{x}.{x}\equiv \la{y}.{y}$, $\la {x'}.yx'{\equiv}\la {y'}.yy'$ and   $\la {x'}.yx'{\equiv}\la {z}.yz$, etc.
Note that
$\la{x}.{xy}\not=_{\cal M} \la{z}.{xz}$, 
 $\la{x}.{x}\not=_{\cal M} \la{y}.{y}$, etc. 

\subsection{Defining  $=_\alpha$ using  the replacement based on ordered variables}
\label{alpharedsecction}

\begin{definition}  \label{alphared}
We define $\alre$\index{reduction ! alpha}\index{$\alre$} as the least compatible relation closed under:
\[ (\alpha) \hspace{0.5in} \la{v}.A \rightarrow_\alpha \la{v'}.A\langle\langle v:=v'\rangle\rangle  \hspace{0.5in} \mbox{where $v' \not \in FV(A)$}
\]
We call this reduction relation $\alpha$-reduction, $ \la{v}.A$ an $\alpha$-redex and\break $\la{v'}.A\langle\langle v:=v'\rangle\rangle$
its $\alpha$-contractum.  If $R$ is an $\alpha$-redex, we write $\Gamma_\alpha[R]$ for its $\alpha$-contractum.
We define $\rrightarrow_\alpha$ (resp.\ $=_\alpha$) as the reflexive transitive (resp.\ equivalence) closure of $\rightarrow_\alpha$.
\end{definition}

Note that $\rightarrow_\alpha$ is not symmetric.
E.g., using the variable list of page~\pageref{orderedlist}:
\begin{itemize}
  \item
    $\la xy.xy \rightarrow_\alpha \la y.(\la y.xy)\langle\langle x:=y\rangle\rangle =_{\cal M}$\\$  \la y.\la z.(xy)\langle\langle y:=z\rangle\rangle\langle\langle x:=y\rangle\rangle =_{\cal M} \la yz.yz$.
  \item
    $\la yz.yz \rightarrow_\alpha \la x.(\la z.yz)\langle\langle y:=x\rangle\rangle =_{\cal M} \la xz.xz\not =_{\cal M} \la xy.xy$.
  \item
    So, $\la xy.xy \rightarrow_\alpha \la yz.yz$ but $\la yz.yz \not \rightarrow_\alpha    \la xy.xy$.\\
    However, $\la yz.yz \rrightarrow_\alpha \la xy.xy$.
\end{itemize}
In fact, $\rrightarrow_\alpha$ is symmetric and hence $=_\alpha$ is the same relation as $\rrightarrow_\alpha$.  See Lemma~\ref{lemmaalphatranslate}.\ref{rrightarrowalphasymmetric}.

Note also that $=_\alpha$ is closed under replacement (Definition~\ref{orderedsubstdef}),  and denotational meaning (Definition~\ref{meaningfuncdef}) and we can now remove the bound variable conditions in
Lemma~\ref{orderedmeta-sublema}.\ref{orderedmeta-sublema3} as long as we use $=_{\alpha}$ instead of 
$=_{\cal M}$.
\begin{lemma}
  \label{eqalphalemmaden}
  \label{lemmaalphapreservesfreevars}
  \begin{enumerate}
  \item
    \label{lemmaalphapreservesfreevars0}
    If $A=_\alpha B$ then $FV(A)=FV(B)$ and  $\den{A} = \den{B}$.
    \item
    \label{lemmaalphapreservesfreevars1}  
  If $v\not = v'$, $v'\not\in FV(A)$ then $A\langle\langle v:=v'\rangle\rangle \langle\langle v':=B\rangle\rangle  =_{\alpha} A\langle\langle v:=B\rangle\rangle  $.    
\item
  \label{lemmaalphapreservesfreevars2}
      If $A=_\alpha B$ then  $C\langle\langle v:=A \rangle\rangle =_\alpha C\langle\langle v:=B \rangle\rangle$ and $A\langle\langle v:=C \rangle\rangle =_\alpha B\langle\langle v:=C \rangle\rangle$.
    \item
       \label{lemmaalphapreservesfreevars3}
      If $v\not = v'$ and $v\not\in FV(C)$, then $A\langle\langle v:=B\rangle\rangle \langle\langle v':=C\rangle\rangle  =_{\alpha} A\langle\langle v':=C\rangle\rangle \langle\langle v:=B\langle\langle v':=C\rangle\rangle \rangle\rangle $.
      \end{enumerate}
\end{lemma}
\begin{proof}
  \ref{lemmaalphapreservesfreevars0}.\ is by induction on the derivation $A=_\alpha B$ using Lemma~\ref{orderedmeta-sublema}.(\ref{orderedlemmasubstnofreevars} and~\ref{orderedlemmafreevarsofasubst}) for the case $A =_{\cal M} \la v.A' \rightarrow_\alpha \la v'.A'\langle\langle v:=v'\rangle\rangle =_{\cal M} B$ where $v'\not\in FV(A')$.\\
  \ref{lemmaalphapreservesfreevars1}.\ is by induction on $A$ using Lemma~\ref{orderedmeta-sublema}.(\ref{orderedlemmasubstnofreevars} and~\ref{orderedlemmafreevarsofasubst}). This also involves a number of sublemmas, all are proved by straightforward induction (see \cite{curry/feys}).\\
  \ref{lemmaalphapreservesfreevars2}.\
  The proof of $C\langle\langle v:=A \rangle\rangle =_\alpha C\langle\langle v:=B \rangle\rangle$ is by induction on $\#A$.
  The proof of $A\langle\langle v:=C \rangle\rangle =_\alpha B\langle\langle v:=C \rangle\rangle$ is by induction on the derivation  $A=_\alpha B$.  Both proofs use \ref{lemmaalphapreservesfreevars0}.\ above.\\
    \ref{lemmaalphapreservesfreevars3}.\ By induction on  $\#A$ using Lemma~\ref{orderedmeta-sublema}.(\ref{orderedlemmasubstnofreevars} and~\ref{orderedlemmafreevarsofasubst}).  This also involves a number of sublemmas, all are proved by straightforward induction.
\end{proof}
  \subsection{Defining $=_{\alpha'}$}
  In Section~\ref{compsec} we stated that the use of a replacement/substitution via an  ordered list of variables suffers from complications and in Lemma~\ref{eqalphalemmaden} we saw that 
we can remove the bound variables conditions in 
Lemma~\ref{orderedmeta-sublema}.\ref{orderedmeta-sublema3} as long as we use $=_{\alpha}$ instead of  $=_{\cal M}$.  But $=_\alpha$ (Definition~\ref{alphared}) is still built on the replacement/substitution via an  ordered list of variables (Definition \ref{orderedsubstdef}).

 Can we forget completely about the ordered variable list  and the replacement notion using the ordered variable list (step 1 in the above approach), and simply use grafting (Definition~\ref{newgraftingdef}) to define an alternative to $\alpha$-conversion that we can use to define  $\beta$-reduction?

  The next definition attempts to introduce syntactic equivalence $=_{\alpha'}$ which uses grafting but will coincide with the $\alpha$-congruence (which is based on the replacement which uses an ordered variable list).   The use of grafting  here will not cause problems since we are applying grafting in a well controlled situation which is guaranteed by the preconditions of the $ ({\alpha'})$ rule (recall Lemma~ \ref{orderedmeta-sublema}.\ref{orderedmeta-sublema4}).

  \begin{definition}
   \label{syntactiequidef}
        We define  $\rightarrow_{\alpha'}$ as the compatible closure of the following rule:
  $$ ({\alpha'}) \hspace{0.1in}  \la v.A \rightarrow_{\alpha'} \la v'.A\{v:=v'\} \hspace{0.2in}\mbox{ if  $v'\not\in FV(vA)$ and $ v,v'\not\in BV(A) $}$$
 
   We define $\rrightarrow_{\alpha'}$ (resp.\ $=_{\alpha'}$) as the reflexive transitive (resp.\ equivalence) closure of $\rightarrow_{\alpha'}$.
  \end{definition}
  
  \begin{lemma}
    \label{alpha0issymmetric}
  If  $v'\not\in FV(vA)$ and $ v,v'\not\in BV(A) $ then $A\{v:=v'\}\{v':=v\} = A$ and hence $\rightarrow_{\alpha'}$ (resp.\ $\rrightarrow_{\alpha'}$) is symmetric and so, $\rrightarrow_{\alpha'}$ is the same relation as $=_{\alpha'}$ on $\la$-terms.
  \end{lemma}
  \begin{proof}
    The first part is by induction on $A$. Symmetry of $\rightarrow_{\alpha'}$ is by induction on $\rightarrow_{\alpha'}$ using what we just proved.  Symmetry of $\rrightarrow_{\alpha'}$ then follows.
    \end{proof}
  Here is a lemma that establishes that using grafting inside the $\alpha'$ rule is safe and that $=_\alpha$ is the same relation as $=_{\alpha'}$.  
 \begin{lemma}
   \label{lemmaalphatranslate}
   \begin{enumerate}
      \item
        \label{lemmaalphatranslate0}
        If $A\rightarrow_{\alpha'} B$ then $\#A = \#B$, $FV(A) = FV(B)$ and if $A = \la v.A'$ then $B = \la v'.B'$.
   \item
     \label{lemmaalphatranslate1}
     If  $v'\not\in FV(vA)$ and $ v,v'\not\in BV(A) $ then 
     $A\{v:=v'\} =_{\cal M} A\langle\langle v:=v'\rangle\rangle$.
   \item
     \label{lemmaalphatranslate2}
     If $A \rightarrow_{\alpha'} B$  then $A \rightarrow_{\alpha} B$.  Hence if $A \rrightarrow_{\alpha'} B$ 
     (resp.\ $A=_{\alpha'} B$) then $A \rrightarrow_{\alpha} B$ (resp.\ $A =_\alpha B$).
   \item
     \label{lemmaalphatranslate3}
     For any $A$, $v_1, v_2, \cdots v_n$, we can find $A'$ such that $A\rrightarrow_{\alpha'} A'$ and $BV(A')\cap \{v_1, v_2, \cdots v_n\} = \emptyset$.
   \item
     \label{lemmaalphatranslate4}
     For any $A$, $v$ and $v'\not\in FV(vA)$, there is $A'$ such that $v'\not\in FV(vA')$, $v,v'\not\in BV(A')$, $A \rrightarrow_{\alpha'} A'$ and $A\langle\langle v:=v'\rangle\rangle \rrightarrow_{\alpha'} A'\langle\langle v:=v'\rangle\rangle$.
   \item
     \label{lemmaalphatranslate5}
     If $A \rightarrow_{\alpha} B$  then $A \rrightarrow_{\alpha'} B$.   Hence if $A \rrightarrow_{\alpha} B$ (resp.\ $A=_{\alpha} B$) then $A \rrightarrow_{\alpha'} B$ (resp.\ $A =_{\alpha'} B$).
     \item
       \label{rrightarrowalphasamealphap}
     $=_\alpha$ is the same relation as $=_{\alpha'}$
     \item
       \label{rrightarrowalphasymmetric}
     $\rrightarrow_\alpha$ is symmetric.
   \end{enumerate}
 \end{lemma}
 \begin{proof}
   \begin{enumerate}
   \item
     Easy induction.
   \item
     By induction on $A$.
   \item
     By induction on the derivation $A\rightarrow_{\alpha'} B$ resp.\  $A\rrightarrow_{\alpha'} B$
     resp.\ $A=_{\alpha'} B$.
   \item
     By induction on $A$. We only do the case $A =_{\cal M} \la v.B$.  By induction hypothesis, there is $B'$ such that $B\rrightarrow_{\alpha'} B'$ and $BV(B')\cap (\{v_1, v_2, \cdots v_n\}\cup\{v\}) = \emptyset$. Let $v'\not\in (FV(vB')\cup \{v_1, v_2, \cdots v_n\}\cup BV(B'))$.  Then, $A =_{\cal M} \la v.B  \rrightarrow_{\alpha'} \la v.B'  \rightarrow_{\alpha'} \la v'.B'\{v:=v'\}$.  Let  $A'=_{\cal M} \la v'.B'\{v:=v'\}$.  Now, 
     $BV(A')\cap \{v_1, v_2, \cdots, v_n\} = (\{v'\}\cup BV(B')) \cap \{v_1, v_2, \cdots, v_n\} = (\{v'\}\cap \{v_1, v_2, \cdots, v_n\})
     \cup (BV(B') \cap \{v_1, v_2, \cdots, v_n\})=\emptyset$. 
   \item
     By induction on $\# A$.   We only do the case $A = \la v''.B$.
     \begin{itemize}
     \item
       Case $v''\not\in \{v,v'\}$ then by the induction hypothesis (IH), there is $B'$ such that  $v'\not\in FV(vB')$, $v,v'\not\in BV(B')$, $B \rrightarrow_{\alpha'} B'$ and $B\langle\langle v:=v'\rangle\rangle \rrightarrow_{\alpha'} B'\langle\langle v:=v'\rangle\rangle$.  Hence  $v'\not\in FV(v(\la v''.B'))$, $v,v'\not\in BV(\la v''.B')$, $\la v''. B \rrightarrow_{\alpha'} \la v''. B'$ and $\la v''. B\langle\langle v:=v'\rangle\rangle \rrightarrow_{\alpha'} \la v''.B'\langle\langle v:=v'\rangle\rangle$. But $(\la v''. B)\langle\langle v:=v'\rangle\rangle =_{\cal M} \la v''. B\langle\langle v:=v'\rangle\rangle$ and $(\la v''. B')\langle\langle v:=v'\rangle\rangle =_{\cal M} \la v''. B'\langle\langle v:=v'\rangle\rangle$.  We are done.
\item
  Case $v''=v$ (i.e., $A =_{\cal M} \la v.B$) then by \ref{lemmaalphatranslate3}.\ and  \ref{lemmaalphatranslate0}.\ above, there is $\la v_1. B'$ such that $\la v. B \rrightarrow_{\alpha'} \la v_1. B'$, $\#B = \#B'$, $FV(\la v. B)=FV(\la v_1. B')$ and $BV(\la v_1.B') \cap \{v,v'\}= \emptyset$.  
  By IH, let $C$ be such that 
  $B' \rrightarrow_{\alpha'} C$, $B'\langle\langle v:=v'\rangle\rangle \rrightarrow_{\alpha'} C\langle\langle v:=v'\rangle\rangle$, $v'\not\in FV(C)$ and $v,v'\not\in BV(C)$.  Hence
  $\la v. B \rrightarrow_{\alpha'}\la v_1. B' \rrightarrow_{\alpha'} \la v_1. C$, $v'\not\in FV(v(\la v_1.C))$, $v,v'\not\in BV(\la v_1.C)$, and
  since $v\not\in FV(\la v_1. C)$, then
  $(\la v. B)\langle\langle v:=v'\rangle\rangle =_{\cal M} \la v.B \rrightarrow_{\alpha'} \la v_1. C=_{\cal M} ( \la v_1. C)\langle\langle v:=v'\rangle\rangle$.
\item
  Case $v'' = v'$  (i.e., $A =_{\cal M} \la v'.B$) and $v\not\in FV(B)$ then by \ref{lemmaalphatranslate3}.\ and  \ref{lemmaalphatranslate0}.\ above, there is $\la v_1. B'$ such that $\la v'. B \rrightarrow_{\alpha'} \la v_1. B'$, $\#B = \#B'$, $FV(\la v'. B)=FV(\la v_1. B')$ and $BV(\la v_1.B') \cap \{v,v'\}= \emptyset$. Since $v'\not\in FV(B') \subseteq (FV(B)\setminus  \{v'\})\cup\{v_1\}$, then
  by IH, there is $C$ such that  
  $B' \rrightarrow_{\alpha'} C$, $B'\langle\langle v:=v'\rangle\rangle \rrightarrow_{\alpha'} C\langle\langle v:=v'\rangle\rangle$, $v'\not\in FV(C)$ and $v,v'\not\in BV(C)$  and by \ref{lemmaalphatranslate0}.\ above, $v\not\in FV(C)$.  Hence
  $\la v'. B \rrightarrow_{\alpha'}\la v_1. B' \rrightarrow_{\alpha'} \la v_1. C$, $v'\not\in FV(v(\la v_1.C))$, $v,v'\not\in BV(\la v_1.C)$, and
  since $v\not\in FV(v'v_1BC)$, 
  $(\la v'. B)\langle\langle v:=v'\rangle\rangle =_{\cal M} \la v'.B \rrightarrow_{\alpha'} \la v_1. C=_{\cal M} ( \la v_1. C)\langle\langle v:=v'\rangle\rangle$.
  
  \item
    Case $v'' = v'$  (i.e., $A =_{\cal M} \la v'.B$) and $v\in FV(B)$
    then \\ $(\la v'.B)\langle\langle v:=v'\rangle\rangle =_{\cal M} \la v''_1.B\langle\langle v':=v''_1\rangle\rangle\langle\langle v:=v'\rangle\rangle$
    where
    $v''_1$ is  the first variable such that $v''_1\not\in FV(v'B)$ (and so, $v''_1 \not = v$). By Lemma~\ref{orderedmeta-sublema}.(\ref{orderedlemmasubstnofreevars} and~\ref{orderedlemmafreevarsofasubst}), $v'\not\in FV(B\langle\langle v':=v''_1\rangle\rangle $.    
    By IH, there is $C$ such that $v'\not\in FV(C)$, $v,v'\not\in BV(C)$,
    $B\langle\langle v':=v''_1\rangle\rangle \rrightarrow_{\alpha'} C$, \\$\la v''_1.B\langle\langle v':=v''_1\rangle\rangle \rrightarrow_{\alpha'} \la v''_1.C$
    and\\ $B\langle\langle v':=v''_1\rangle\rangle\langle\langle v:=v'\rangle\rangle \rrightarrow_{\alpha'} C\langle\langle v:=v'\rangle\rangle$.

    Hence $\la v''_1.B\langle\langle v':=v''_1\rangle\rangle\langle\langle v:=v'\rangle\rangle \rrightarrow_{\alpha'} \la v''_1.C\langle\langle v:=v'\rangle\rangle$ and so,  $(\la v'.B)\langle\langle v:=v'\rangle\rangle \rrightarrow_{\alpha'} (\la v''_1.C)\langle\langle v:=v'\rangle\rangle$.

    All that is left now is to show that  $\la v'.B \rrightarrow_{\alpha'} \la v''_1.C$.

    Since $v''_1\not\in FV(v'B)$ then by IH, there is $D$ such that  $v',v''_1\not\in BV(D)$, $v''_1\not\in FV(v'D)$, $B \rrightarrow_{\alpha'} D$,  and 
    $B\langle\langle v':=v''_1\rangle\rangle\rrightarrow_{\alpha'} D\langle\langle v':=v''_1\rangle\rangle$.  Hence, 
    $\la v'.B \rrightarrow_{\alpha'} \la v'. D$, and \\
    $ \la v''_1.B\langle\langle v':=v''_1\rangle\rangle\rrightarrow_{\alpha'}  \la v''_1.D\langle\langle v':=v''_1\rangle\rangle$.\\
    Since  $v',v''_1\not\in BV(D)$, $v''_1\not\in FV(v'D)$, then  by ~\ref{lemmaalphatranslate1}.\ above,  \\
    $\la v'.D \rightarrow_{\alpha'} \la v''_1. D\{v':=v''_1\} =_{\cal M} \la v''_1. D\langle\langle v':=v''_1\rangle\rangle$.\\
    Since $\rrightarrow_{\alpha'}$ is symmetric (Lemma~\ref{alpha0issymmetric}), then\\
    $  \la v''_1.D\langle\langle v':=v''_1\rangle\rangle \rrightarrow_{\alpha'} 
    \la v''_1.B\langle\langle v':=v''_1\rangle\rangle$. Hence:\\
    $\la v'B. \rrightarrow_{\alpha'} \la v'. D  \rightarrow_{\alpha'}  \la v''_1. D\langle\langle v':=v''_1\rangle\rangle \rrightarrow_{\alpha'} 
    \la v''_1.B\langle\langle v':=v''_1\rangle\rangle \rrightarrow_{\alpha'} \la v''_1.C$ and we are done.
       
       \end{itemize}
   \item
     The proof of $ \rightarrow_{\alpha} \subseteq \rrightarrow_{\alpha'}$ is by induction on the derivation $A \rightarrow_{\alpha} B$. We only do one case.\\
     Assume $\la v.C  \rightarrow_{\alpha} \la v'.C\langle\langle v:=v'\rangle \rangle$ where $v'\not\in FV(vC)$ (the case $v = v'$ is trivial). Then by \ref{lemmaalphatranslate4}.\ above,  there is $C'$ such that $v'\not\in FV(vC')$, $v,v'\not\in BV(C')$, $C \rrightarrow_{\alpha'} C'$ and $C\langle\langle v:=v'\rangle\rangle \rrightarrow_{\alpha'} C'\langle\langle v:=v'\rangle\rangle$.  Hence  $\la v.C \rrightarrow_{\alpha'} \la v.C'$ and $\la v'.C\langle\langle v:=v'\rangle\rangle \rrightarrow_{\alpha'} \la v'.C'\langle\langle v:=v'\rangle\rangle$. By symmetricity Lemma~\ref{alpha0issymmetric}, $\la v'.C'\langle\langle v:=v'\rangle\rangle \rrightarrow_{\alpha'} \la v'.C\langle\langle v:=v'\rangle\rangle$.
     By ~\ref{lemmaalphatranslate1}.\ above,  $\la v'.C'\{ v:=v'\}=_{\cal M}\la v'.C'\langle\langle v:=v'\rangle\rangle
     $.\\
     Hence
     $\la v.C \rrightarrow_{\alpha'} \la v.C' \rightarrow_{\alpha'} \la v'.C'\langle\langle v:=v'\rangle\rangle \rrightarrow_{\alpha'} \la v'.C\langle\langle v:=v'\rangle\rangle$.
     \item
       Use \ref{lemmaalphatranslate2} and \ref{lemmaalphatranslate5} above.
     \item
       If $A\rrightarrow_\alpha B$ then by \ref{lemmaalphatranslate5}.\ above,  $A\rrightarrow_{\alpha'} B$ and  by  symmetric Lemma~\ref{alpha0issymmetric}, $B\rrightarrow_{\alpha'} A$ which mean by  \ref{lemmaalphatranslate2}.\ above,  $B\rrightarrow_\alpha A$.
      \end{enumerate} 
 \end{proof}
  
  Now that $=_\alpha$ is the same as $=_{\alpha'}$, we use $\equiv$ to denote these relations.

  \begin{definition}[$\equiv$]
    \label{equivdef}
    We write $A\equiv B$ iff $A =_{\alpha'} B$ iff $A =_{\alpha} B$.\\
    When $A\equiv B$, we say that $A$ and $B$ are syntactically equivalent.
   \end{definition}

 \section{Beta Reduction}
 \label{secbetaalpha}
 So far, $\be^w$-reduction of Table~\ref{graftingnogood} does not work and $\overline{\be}$-reduction of Definition~\ref{overlinebetared} does not satisfy CR. But, we are ready to define the computation that will work and will guarantee CR.
 
 \subsection{Computation ${\beal}$-reduction based on $=_\alpha$}
 In this section we define the $\be$-reduction relation (called here ${\beal}$) given in~\cite{hindley/seldin:LCCI}.
 Note that
 although $\rightarrow_{{\beal}}$ is the same as  $\rightarrow_{\overline{\be}}$, a ${\beal}$-redex (resp.\ ${\beal}$-contractum)
 is also a $\overline{\be}$-redex (resp.\ ${\overline{\be}}$-contractum)
 and vice-versa.  However,
 the reflexive transitive closure $\rrightarrow_{\beal}$ incorporates also $\alpha$-reduction (unlike $\rrightarrow_{\overline{\be}}$).
  \begin{definition}  \label{betared}
    We define $\rightarrow_{\beal}$ as $\rightarrow_{\overline{\be}}$ and similarly define a ${\beal}$-redex, a ${\beal}$-contractum and $\Gamma_{\beal}[R]$ when $R$ is a $\beal$-redex.

We define $\rrightarrow_{\beal}$  as the reflexive transitive  closure of $\rightarrow_{\beal}\cup \rightarrow_\alpha$.
  \end{definition}
  Unlike $\rrightarrow_{\overline{\beta}}$ (see Example~\ref{ecampledoesnotwork}), 
  this $\rrightarrow_{\beal}$ relation satisfies the CR property (see  \cite{hindley/seldin:ICLC}).

  The next help lemma which is based on lemmas 1.12, 1.13 and 1.14 of
  \cite{hindley/seldin:LCCI} establishes the closure of $\alpha$-congruence under ${\beal}$-reduction.  
  \begin{lemma}
    \label{redexesinsideforCR}
     \begin{enumerate}
     \item
       \label{redexesinsideforCR1}
        If $R$ is a ${\beal}$-redex and $v$, $A$ are such that $FV(vA)\cap BV(R)=\emptyset$ then $R\langle\langle v:=A \rangle\rangle $ is a ${\beal}$-redex and\\         $\Gamma_{\beal}[ R\langle\langle v:=A\rangle\rangle] =_\alpha (\Gamma_{\beal}[R])\langle\langle v:=A\rangle\rangle$.
      \item
        \label{redexesinsideforCR2}
        If $R$ is a ${\beal}$-redex and $R =_\alpha R'$ then\\ $R'$ is also a ${\beal}$-redex and $\Gamma_{\beal}[R] =_\alpha \Gamma_{\beal}[R']$.
          \item
        \label{redexesinsideforCR3}
       If $A =_\alpha B$ and ${\beal}$-redex $R$ occurs in $A$,  then a ${\beal}$-redex $R'$ occurs in $B$ such that if $A \rightarrow_{\beal} A'$ using ${\beal}$-redex $R$ and $B \rightarrow_{\beal} B'$ using ${\beal}$-redex $R'$, then $A' =_\alpha B'$.
                \end{enumerate}
        \end{lemma}
  \begin{proof}
    \begin{itemize}
    \item[\ref{redexesinsideforCR1}.]
      Assume $R =_{\cal M} (\la v'.B)C$.  Then, $v'\not = v$ and $v'\not\in FV(A)$.  Hence, 
    $R\langle\langle v:=A\rangle\rangle =_{\cal M} (\la v'.B)\langle\langle v:=A\rangle\rangle C\langle\langle v:=A\rangle\rangle =_{\cal M} (\la v'.B\langle\langle v:=A\rangle\rangle)C\langle\langle v:=A\rangle\rangle$ is a ${\beal}$-redex.  Moreover,     \\     $\Gamma_{\beal} [R\langle\langle v:=A\rangle\rangle] =_\alpha B\langle\langle v:=A\rangle\rangle \langle\langle v':=C\langle\langle v:=A\rangle\rangle\rangle\rangle =_{\alpha}^{Lemma~\ref{orderedmeta-sublema}.\ref{orderedmeta-sublema3} } B\langle\langle v':=C\rangle\rangle\langle\langle v:=A\rangle\rangle =_\alpha 
      (\Gamma_{\beal} [R])\langle\langle v:=A\rangle\rangle$.
    \item[\ref{redexesinsideforCR2}.]
      By induction on the derivation $R=_\alpha R'$ using Lemmas~\ref{eqalphalemmaden}.(\ref{lemmaalphapreservesfreevars1}, \ref{lemmaalphapreservesfreevars2}).\\
    \item[      \ref{redexesinsideforCR3}.] See \cite{hindley/seldin:LCCI}, Lemma A1.14 (b).    
      \end{itemize}
    \end{proof}
  In the next lemma we connect both the ${\beal}$ and $\overline{\be}$ relations.
  \begin{lemma}\label{betaandbetaoverareconnected}
    \begin{enumerate}
    \item
      \label{betaandbetaoverareconnected1}
        $A \rightarrow_{\beal} B$ iff  $A \rightarrow_{\overline{\beta}} B$.
    \item
      \label{betaandbetaoverareconnected2}
        If   $A \rrightarrow_{\overline{\beta}} B$ then $A \rrightarrow_{\beal} B$.
      \item
        \label{betaandbetaoverareconnected3}
        If $A \rrightarrow_{\beal} B$ then there is $C$  such that $A\rrightarrow_{\overline{\beta}} C\rrightarrow_{\alpha} B$.
        \end{enumerate}
    \end{lemma}
  \begin{proof}
    \ref{betaandbetaoverareconnected1}.\ is obvious since $({\beal})$ and $({\overline{\beta}})$ rules are the same.\\
    \ref{betaandbetaoverareconnected2}.\ is a corollary of 1.\\
    \ref{betaandbetaoverareconnected3}.\ By definition of $\rrightarrow_{\beal}$, $A \rrightarrow_{\beal} B$ is a sequence (possibly empty) of $\rightarrow_{\beal} $ and $\rightarrow_\alpha $ steps.  We will show that any $A_1 \rightarrow_\alpha A_2 \rightarrow_{\beal} A_3$ can be written as $A_1 \rightarrow_{\beal} A_4 \rrightarrow_\alpha A_3$  for some $A_4$.

    Since  $A_1 \rightarrow_\alpha A_2 \rightarrow_\be A_3$ then $A_1 =_\alpha A_2$ and $A_2 \rightarrow_{\beal} A_3$ (say by ${\beal}$-redex $R$).  By Lemma~\ref{redexesinsideforCR}.\ref{redexesinsideforCR3}, there is a ${\beal}$-redex $R'$ in $A_1$
    such that
    $A_1 \rightarrow_{\beal} A_4$ using ${\beal}$-redex $R'$ and $A_4 =_\alpha A_3$.  But $=_\alpha$ is the same as relation as $\rrightarrow_\alpha$ and therefore,
    $A_1 \rightarrow_{\beal} A_4 \rrightarrow_\alpha A_3$.

    This means all the $\rightarrow_\alpha$-steps 
    can be postponed till after all the $\rightarrow_{\beal}$-steps.  By what we just proved and \ref{betaandbetaoverareconnected1} above,
    if $A \rrightarrow_{\beal} B$ then there is $C$  such that $A\rrightarrow_{\overline{\beta}} C\rrightarrow_{\alpha} B$.
  \end{proof}

  \subsection{${\beall}$-reduction based on clean terms and $=_{\alpha'}$}
   In this section we define the $\beall$-reduction relation. The basic ($\beall$) rule first transforms the redex into a so-called clean term before any beta reduction takes place and uses grafting since clean terms are safe with grafting.   Here, the reflexive transitive closure $\rrightarrow_{\beall}$ also incorporates  $\alpha'$-reduction.  So, for grafting to work, we use   $\alpha'$-reduction at the basic reduction stage (to clean the term) and at the reflexive transitive stage.

  \begin{definition}
    \label{defclean}
    A $\la$-term $A$ is clean iff the following two conditions hold:
    \begin{itemize}
    \item
      $BV(A)\cap FV(A) = \emptyset$.
    \item
            For any $v$,  $\la v$ may occur at most once in $A$. 
      \end{itemize}
  \end{definition}
  So $\la x.\la y.(\la z.xz(yz))(\la y.yz)$ is not clean.  However, a clean version is 
  $\la x.\la y.(\la z'.xz'(yz'))(\la y'.y'z)$.
  Of course a term may have different clean versions.  E.g., $\la x'.\la y.(\la z'.x'z'(yz'))(\la y'.y'z)$ is also a (different)  clean version of  the term above.
  \begin{lemma}
    \label{cleantermlem}
    For any $A$, there is a clean $B$ such that $A\rrightarrow_{\alpha'} B$.
  \end{lemma}
  \begin{proof}
    By Lemma~\ref{lemmaalphatranslate}.\ref{lemmaalphatranslate3}, we can find $A'$ such that $A\rrightarrow_{\alpha'} A'$ and $BV(A')\cap (FV(A)\cup BV(A))    = \emptyset$. By Lemma~\ref{lemmaalphatranslate}.\ref{lemmaalphatranslate0}, $FV(A) = FV(A')$. Hence,  $BV(A')\cap FV(A')   = \emptyset$.  If $BV(A') = \emptyset$, we are done. Take $B$ to be $A'$. 
    
    Else, assume that $A'$ has $n$ $\la$s $\la v_1, \la v_2, \cdots, \la v_n$ (note that some of the $v_i$s may be equal)
    occurring from left to right in that order.    
    Then
    $A' =_{\cal M} C_n[\la v_n.B_n]$ for some context $C_n[\:]$ and some term $B_n$. Let $v'_n\not\in FV(A')\cup BV(A')$.  Hence,  $v'_n\not\in FV(v_nB_n)\cup BV(B_n)$ and $A'\rightarrow_{\alpha'}  C_n[\la v'_n.B_n\{v_n:=v'_n\}] =_{\cal M} A'_n$.  Note that
    $FV(A') = FV(A'_n)$, $v'_n\not\in FV(A'_n)$ and $A'_n$ has $n$ $\la$s $\la v_1, \la v_2, \cdots,\la v_{n-1}, \la v'_n$ 
    occurring from left to right in that order, where $\la v'_n$ occurs  once in $A'_n$ and $\{v_1, v_2, \cdots, v_{n-1},\break  v'_n\}\cap FV(A'_n) = \emptyset$. \\
    Repeat the above process for $\la v_{n-1}$ where   $A'_n =_{\cal M} C_{n-1}[\la v_{n-1}.B_{n-1}]$ for some context $C_{n-1}$ and term $B_{n-1}$. Let $v'_{n-1}\not\in FV(A'_n)\cup BV(A'_n)$.  
    Hence,  $v'_{n-1}\not\in FV(v_{n-1}B_{n-1})\cup BV(B_{n-1})$ and $A'_n\rightarrow_{\alpha'}  C_{n-1}\break [\la v'_{n-1}.B_{n-1}\{v_{n-1}:=v'_{n-1}\}] =_{\cal M} A'_{n-1}$.  Note that
    $FV(A'_n) =\break FV(A'_{n-1})$, $v'_{n-1}\not\in FV(A'_{n-1})$ and $A'_{n-1}$ has $n$ $\la$s $\la v_1$, $\la v_2$,  $\cdots$, $\la v_{n-2}$, $\la v'_{n-1}, \la v'_n$ 
    occurring from left to right in that order, where $v_n \not = v_{n-1}$ and each of $\la v'_n$ and $\la v'_{n-1}$ occurs once in $A'_{n-1}$ and $\{v_1, v_2, \cdots, v_{n-2},\break  v'_{n-1},  v'_n\}\cap FV(A'_{n-1}) = \emptyset$. \\
    Like we constructed $A'_n$, $A'_{n-1}$, we continue to construct $A'_{n-2}$, $\cdots$, $A'_2$, $A'_1$ such that $A' \rightarrow_{\alpha'} A'_n  \rightarrow_{\alpha'} A'_{n-1}  \rightarrow_{\alpha'} A'_{n-2} \cdots  \rightarrow_{\alpha'} A'_1$.\\  $A'_1$ is the clean term we are after.
  \end{proof}
  \begin{lemma}
    \label{cleangraftsubs}
    If $(\la v.A)B$ is clean then $A\{v:=B\} =_{\cal M} A\langle\langle v':=B\rangle\rangle$.
  \end{lemma}
  \begin{proof}
    By induction on $A$.
  \end{proof}
  Based on Lemmas~\ref{cleantermlem} and~\ref{cleangraftsubs}, we can define beta reduction based on
  clean terms and $\alpha'$-reduction as follows:
  \begin{definition}  \label{betasecondclean}
We define $\rightarrow_{\beall}$ as the least compatible relation closed under:
\[ \begin{array}{lll}
  ({\beall}) \hspace{0.1in} &(\la{v}.A)B \rightarrow_{\beall} A'\{ v':=B'\} \\
  &  \mbox{where $(\la v.A)B \rrightarrow_{\alpha'} (\la v'.A')B'$ and $(\la v'.A')B'$ is clean.} 
  \end{array}\]
We call this reduction relation ${\beall}$-reduction.
We define $\rrightarrow_{\beall}$  as the reflexive transitive closure of $\rightarrow_{\beall}\cup \rightarrow_{\alpha'}$.\footnote{Note that in the above definition of $\rrightarrow_{\beall}$, the reflexive transitive closure of $\rightarrow_{\beall}\cup \rightarrow_{\alpha'}$ is necessary since otherwise,
  if we take  $\rrightarrow_{\beall}$  as the reflexive transitive closure of $\rightarrow_{\beall}$ then we would lose CR.  For example, $(\la x.(\la x.x)x)x\rightarrow_{\beall} x(\la y.y)$ and
  $(\la x.(\la x.x)x)x\rightarrow_{\beall} x(\la z.z)$ but in the absence of $\rightarrow_{\alpha'}$ we can never show that $ x(\la y.y)$ and $ x(\la z.z)$ ${\beall}$-reduce to a common term (say $x(\la u.u)$).}
  \end{definition}
  \begin{lemma}
    \label{detasecondsameasbetap}
    \begin{enumerate}
    \item
      If $A \rightarrow_{\beall} B$ then $A \rrightarrow_{\beal} B$.
    \item
        If $A \rightarrow_{\beal} B$ then $A \rrightarrow_{\beall} B$.
      \end{enumerate}
  \end{lemma}
  \begin{proof}
    1.\ By induction on $A \rightarrow_{\beall} B$. For the case $(\la{v}.A)B \rightarrow_{\beall} A'\{ v':=B'\}$ 
    where $(\la v.A)B \rrightarrow_{\alpha'} (\la v'.A')B'$ and $(\la v'.A')B'$ is clean, use Lemmas~\ref{cleangraftsubs} and~\ref{lemmaalphatranslate}.\\
    2.\ \ By induction on $A \rightarrow_{\beal} B$. We do the case $(\la{v}.A)B \rightarrow_{\beal} A\langle\langle v:=B\rangle\rangle$.
    By Lemmas~\ref{cleantermlem} and~\ref{cleangraftsubs}, for a clean $(\la v'.A')B'$,
    $(\la{v}.A)B \rightarrow_{\alpha'} (\la v'.A')B'$ and 
    $(\la{v}.A) B \rightarrow_{\beall} A'\langle\langle v':=B'\rangle\rangle$.
    By Lemma~\ref{redexesinsideforCR}.\ref{redexesinsideforCR2},
    since    $(\la{v}.A)B =_\alpha (\la{v'}.A')B'$ then $A\langle\langle v:=B\rangle\rangle =_{\cal M}\Gamma_{\beal}[(\la{v}.A)B] =_\alpha \Gamma_{\beal}[(\la{v'}.A')B']=_{\cal M} A'\langle\langle v':=B'\rangle\rangle$. Hence by Lemma~\ref{lemmaalphatranslate}, $A'\langle\langle v':=B'\rangle\rangle \rrightarrow_{\alpha'}  A\langle\langle v:=B\rangle\rangle$ and so, $(\la{v}.A)B \rrightarrow_{\beall} A\langle\langle v:=B\rangle\rangle$.
  \end{proof}
  Since ${\beal}$ satisfies CR,
  by Lemma~\ref{detasecondsameasbetap}, also ${\beall}$ satisfies CR.

   However, there is an oddity about $\beall$ compared to $\beal$. The latter is a function whereas the former is not.  
   If $A \rightarrow_{\beal} B$ then $B$ is unique because the replacement relation for $\beal$ is based on an ordered variable list.  This is not the case for $\rightarrow_{\beall}$ where for example,  $(\la x.x(\la x.x))x  \rightarrow_{\beall} x(\la y.y)$ and $(\la x.x(\la x.x))x  \rightarrow_{\beall} x(\la z.z)$ but  $x(\la y.y)\not= x(\la z.z)$. 
   
\subsection{Equating terms modulo $\equiv$}
  \label{secsubst}
  Now that syntactic equivalence $\equiv$ is defined (and is the same relation as $\alpha$-congruence $=_{\alpha}$ set up in terms of replacement using an ordered variable list and  $\alpha'$-congruence $=_{\alpha'}$ set up in terms of grafting), we could
  identify terms modulo syntactic equivalence $\equiv$, and then define substitution as a refined form of both replacement using an ordered variable list and grafting.  For this, all we need to do is to  use  $\equiv$ instead $=_{\cal M}$ in Definition~\ref{orderedsubstdef}
  and to remove the statement ``and $v''$ is the first variable in the ordered list ${\cal V}$'' from clause 6. This gives us the following definition:
 \begin{definition}[Substitution  \mbox{$A [ v:=B]$}, using $\equiv$]
\label{alphasubstdef}
For any $A, B, v$, we define $A{[v:=B] }$ to be the result of substituting $B$ for every free occurrence of $v$ in $A$, as follows:
\[
\begin{array}{llll}
1. \: v{[ v:=B]} & \equiv & B  \\
2. \:  v'{[ v:=B]} & \equiv & v'  \hspace{0.5in} \mbox{if $v \not = v'$}\\
3. \: (A C){[ v:=B]} & \equiv  & A{[ v:=B]} C{[ v:=B]}  \\
4. \: (\la{v}.{A}){[ v:=B]} & \equiv & \la{v}.{A}  \\
5. \: (\la{v'}.{A}){[v:=B]} & \equiv & \la{v'}.{A{[ v:=B]}}    \hspace{0.5in} \mbox{if $v \not = v'$}\\
 & & \mbox{and ($v' \not \in FV(B)$ or $v \not \in FV(A)$)}\\
6. \: (\la{v'}.{A}){[ v:=B]} & \equiv & \la{v''}.{A{[ v':=v'']}}{[ v:=B]}  \hspace{0.3in} \mbox{if $v \not = v'$} \\
 & & \mbox{and ($v' \in FV(B)$ and  $v \in FV(A)$)} \\
  & & \mbox{such that $v''\not\in FV(AB)$.}
\end{array}
\]
 \end{definition}
 Recall Lemma~\ref{eqalphalemmaden}.\ref{lemmaalphapreservesfreevars2} which implies that if $A\equiv B$ and $C\equiv D$ then $A\langle\langle v:=C \rangle\rangle \equiv B\langle\langle v:=D\rangle\rangle$.  
Using $\equiv$ instead of $=_{\cal M}$ means that we can also move from $A\langle\langle v:=B\rangle\rangle$ to $A[v:=B]$ as seen by the following lemma which is proven by induction on $\#A$.
  \begin{lemma}
   \label{substandrepla}
     For any $v$, $A$ and $B$, $A[v:=B]\equiv A\langle\langle v:=B\rangle\rangle$.
   \end{lemma}
  We could even go one step further in our use of $=_\alpha$ and use the so-called variable convention~\cite{Bar82}  where we assume that no variable name is both free and bound within the same term.\footnote{Clean terms satisfy the variable convention, but the other way round does not necessarily hold.}  So, we will never have terms like $(\la v.B)[v:=C]$ or  $(\la v.B)v$.  Instead, the $\la v$ would be changed to a $\la v'$ where $v\not = v'$.
  This way, clauses 4 and 6 of Definition~\ref{alphasubstdef} do not hold and we know that in clause 5,  $v \not = v'$
 and ($v' \not \in FV(B)$ or $v \not \in FV(A)$) always hold.  
  Hence, if we always write terms following the variable convention,  can define substitution as:
   \begin{definition}
\label{alphasubstdefnew}
For any $A, B, v$, we define $A{[v:=B] }$ to be the result of substituting $B$ for every free occurrence of $v$ in $A$, as follows:
\[
\begin{array}{llll}
1. \: v{[ v:=B]} & \equiv & B  \\
2. \:  v'{[ v:=B]} & \equiv  & v'  \hspace{0.5in} \mbox{if $v \not = v'$}\\
3. \: (A C){[ v:=B]} & \equiv  & A{[ v:=B]} C{[ v:=B]}  \\
5. \: (\la{v'}.{A}){[v:=B]} & \equiv  & \la{v'}.{A{[ v:=B]}} \\
\end{array}
\]
\end{definition}

   So here, we always use terms modulo names of bound variables and use $\equiv$ (which is the same as $=_{\alpha}$ and $=_{\alpha'}$).  In this case, we could use the replacement given in Definition~\ref{alphasubstdef}.  If we also want our terms to be written according to the variable convention (where the names of bound variables are always different from the free ones) then instead of Definition~\ref{alphasubstdef}, we can use Definition~\ref{alphasubstdefnew}. Whicever style we take here (always modulo bound variable names, with or without variable convention), we define $\be$-reduction as follows:
   \begin{definition}  \label{betaredwithalpha}
We define $\rightarrow_{\beta}$ as the least compatible relation closed under:
\[ ({\beta}) \hspace{0.5in} (\la{v}.A)B \rightarrow_{\beta} A[ v:=B] \hspace{0.5in} \]
We call this reduction relation $\beta$-reduction.
We define $\rrightarrow_{\beta}$  as the reflexive transitive closure of $\rightarrow_{\beta}$.
\end{definition}
   Let $r\in\{{\overline{\be}}, {\beal}, {\beall}\}$.
   It is easy to show that $\rightarrow_r\subseteq \rightarrow_\be$ and that, if $A\rightarrow_\be B$ then there are $A'$, $B'$ such that  $A'\equiv A$, $B'\equiv B$, and $A'\rightarrow_r B'$.

   \medskip

   We end this section by summarising the substitutions and $\be$-reduc\-tions introduced in this paper.
   The basic axioms introduced are:
 \begin{itemize}
 \item
    $({\be^w})\hspace{0.1in} (\la v.A)B \rightarrow_{\be^w} A\{v:=B\}$.
 \item
   $(\alpha) \hspace{0.1in} \la{v}.A \rightarrow_\alpha \la{v'}.A\langle\langle v:=v'\rangle\rangle  \hspace{0.1in} \mbox{where $v' \not \in FV(A)$}$.
 \item
   $ ({\alpha'}) \hspace{0.05in}  \la v.A \rightarrow_{\alpha'} \la v'.A\{v:=v'\} \hspace{0.03in}\mbox{ if  $v'\not\in FV(vA)$ and $ v,v'\not\in BV(A). $}$
 \item
   $ ({\overline{\beta}}) \hspace{0.1in} (\la{v}.A)B \rightarrow_{\overline{\beta}} A\langle\langle v:=B\rangle\rangle.$
\item
  $ ({\beal}) \hspace{0.1in} (\la{v}.A)B \rightarrow_{\beal} A\langle\langle v:=B\rangle\rangle.$
\item
  $ \begin{array}{lll}
  ({\beall}) \hspace{0.1in} &(\la{v}.A)B \rightarrow_{\beall} A'\{ v':=B'\} \\
  &  \mbox{where $(\la v.A)B \rrightarrow_{\alpha'} (\la v'.A')B'$ and $(\la v'.A')B'$ is clean.}
  \end{array}$
\item
     $({\be})\hspace{0.1in} (\la v.A)B \rightarrow_{\be} A[v:=B]$.
 \end{itemize}
The reduction relations based on these axioms are:
 \begin{itemize}
 \item
   For each $r\in\{{\be^w}, \alpha, \alpha', \overline{\be}, {\beal}, {\beall}, \be\}$:
   \begin{itemize}
   \item
     We define  $\rightarrow_{r}$ as the compatible closure of $(r)$ and
     call this reduction relation $r$-reduction.
        \item
     We call the term on the left (resp.\ right) of  $\rightarrow_r$ in $(r)$, an $r$-redex (resp.\ an $r$-contractum).
   \item
     If $R$ is an $r$-redex, we write $\Gamma_r[R]$ for its  $r$-contractum.
   \end{itemize}
     \item
         For each $r\in\{{\be^w}, \alpha, \alpha', \overline{\be}, \be\}$, we define $\rrightarrow_r$ (resp.\ $=_r$) as the reflexive transitive (resp.\ equivalence) closure of $\rightarrow_r$.
       \item
         We define $\rrightarrow_{\beal}$ as the reflexive transitive (resp.\ equivalence) closure of $\rightarrow_{\beal}\cup\rightarrow_\alpha$.
       \item
         We define $\rrightarrow_{\beall}$  as the reflexive transitive closure of $\rightarrow_{\beall}\cup \rightarrow_{\alpha'}$.
     \end{itemize}

\noindent  \begin{tabular}{|l|l|l|l|p{2.2cm}|l|}
\hline
Replacement & = on &  $(\la v.A)B$ &  Extra &$\rrightarrow_r$ & CR\\
$A_{v:=B}$ & ${\cal M}$ &  $ \rightarrow_r ?$  & & &\\
\hline
$A\{v:=B\}$ & $=_{\cal M}$& $r =\be^w$ \xmark  & &  &\\
& & ? is $A_{v:=B}$ & &&\\
\hline
$A\langle\langle v:=B\rangle\rangle$ &$=_{\cal M}$ & $r = \overline{\be}$&&ref.\ tran.\ clos. & \xmark \\
 & &? is $A_{v:=B}$  & & of $\rightarrow_r$ & \\
\hline
$A\langle\langle v:=B\rangle\rangle$& $=_{\cal M}$ & $r = {\beal} = \overline{\be}$ &$=_{\alpha}$ &ref.\ tran.\ clos. & \checkmark\\
& & ? is $A_{v:=B}$ & & of $\rightarrow_r\cup\rightarrow_\alpha$ & \\
\hline
$A\{v:=B\}$& $=_{\cal M}$ & $r = {\beall}$ & $=_{\alpha'}$&ref.\ tran.\ clos. &\checkmark\\
& & ? is $A'_{v':=B'}$, & & of $\rightarrow_r\cup\rightarrow_\alpha$ & \\
& &  $(\la v.A)B=_{\alpha'}$ & & & \\
& &  $(\la v'.A')B' $ clean & & & \\
\hline
A[v:=B]& $=_{\alpha}$ & $r = \be $&var. &ref.\ tran.\ clos.& \checkmark\\
 & & ? is $A_{v:=B}$ &conv. & of $\rightarrow_r$ & \\
\hline
 \end{tabular}

\section{The $\la$-calculus with de Bruijn indices,\\ towards Explicit Substitutions}
\label{secdeb}
We can  avoid the problem of variable capture by getting rid of variables and using {\em De Bruijn indices} which are natural numbers that represent the occurrences of variables in the term.    In this section, we  introduce the $\la$-calculus with de Bruijn indices and look at the substitution and reduction rules with de Bruijn indices which will lead us naturally to a calculus with explicit substitutions.
\subsection{The classical $\la$-calculus with de Bruijn indices}
In this approach, an index $n$ in a term $A$ is bound  by the $n$th $\la$ on the left of $n$. For example, using de Bruijn indices, 
we write  $\la xy.yx$ and $\la xy.xy$
as follows:
  \begin{center}
  \begin{picture}(70,20)(0,0)
\put(0,12){\line(1,0){30}} 
\put(0,12){\line(0, -1){7}} 
\put(10,10){\line(0, -1){5}} 
\put(20,10){\line(0, -1){5}} 
\put(30,12){\line(0, -1){7}} 
\put(10,10){\line(1,0){10}} 
\put(-5,-5){\makebox(10,10){$\la$}} 
\put(5,-5){\makebox(10,10){${\la}$}} 
\put(15,-5){\makebox(10,10){${1}$}}
\put(25,-5){\makebox(10,10){$2$}} 
\end{picture} 
\begin{picture}(70,20)(0,0)
\put(0,12){\line(1,0){20}} 
\put(0,12){\line(0, -1){7}} 
\put(10,10){\line(0, -1){5}} 
\put(20,12){\line(0, -1){7}} 
\put(30,10){\line(0, -1){5}} 
\put(10,10){\line(1,0){20}} 
\put(-5,-5){\makebox(10,10){$\la$}} 
\put(5,-5){\makebox(10,10){${\la}$}} 
\put(15,-5){\makebox(10,10){${2}$}}
\put(25,-5){\makebox(10,10){$1$}} 
\end{picture}
\end{center}
  If $n$ is free in $A$, then we look in a so-called {\em free variable list}.  Say: \[x, y, z, x', y', z', x'', y'', z'', \dots\]

For example, the trees of $\la x.\la y. zxy$ and its translation $\la \la 5 2 1$ are as follows (here $\delta$ stands for application, the trees are drawn horizontally to save space, and the dashed lines represent the free variable list):
  \begin{center}
\setlength{\unitlength}{3.5mm}
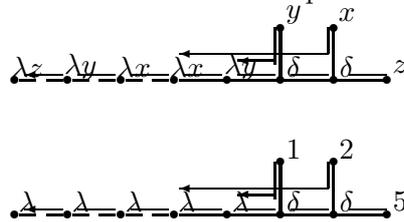
\begin{picture}(20,5)(0,-1)
 \thicklines
\multiput(0.2,2)(1,0){6}{\line(1,0){0.6}}
\put(6.2,2){\line(1,0){2}}
\put(0,2){\circle*{.3}}
 \put(2,2){\circle*{.3}}
 \put(4,2){\circle*{.3}}
 \put(6,2){\circle*{.3}}
 \put(8,2){\circle*{.3}}
 \put(10,2){\circle*{.3}}
 \put(12,2){\circle*{.3}}
 \put(14,2){\circle*{.3}}
 \put(10,4){\circle*{.3}}
 \put(12,4){\circle*{.3}}
 \put(10,2){\line(0,1){2}}
 \put(12,2){\line(0,1){2}}
 \put(8,2){\line(1,0){6}}
 \put(0,2){\makebox(1,1){$\lambda z$}}
 \put(2,2){\makebox(1,1){$\lambda y$}}
 \put(4,2){\makebox(1,1){$\lambda x$}}
 \put(6,2){\makebox(1,1){$\lambda x$}}
 \put(8,2){\makebox(1,1){$\la y$}}
 \put(10,2){\makebox(1,1){$\delta$}}
 \put(12,2){\makebox(1,1){$\delta$}}
 \put(14,2){\makebox(1,1){$z$}}
 \put(10,4){\makebox(1,1){$y$}}
 \put(12,4){\makebox(1,1){$x$}}
\thinlines
\put(14,2.2){\line(-1,0){1.4}}
\put(11.8,2.2){\line(-1,0){1.4}}
\put(9.8,2.2){\line(-1,0){1.4}}
\put(9.8,2.75){\vector(-1,0){1.4}}
\put(7.8,2.2){\line(-1,0){1.4}}
\put(5.8,2.2){\line(-1,0){1.4}}
\put(3.8,2.2){\line(-1,0){1.4}}
\put(1.8,2.2){\vector(-1,0){1.4}}
\put(9.8,4){\line(0,-1){1.4}}
\put(11.8,3){\vector(-1,0){5.6}}
\put(11.8,4){\line(0, -1){1}}
\end{picture}
\vspace{-0.5in}
\begin{picture}(20,5)(0,-1)
 \thicklines
\multiput(0.2,2)(1,0){6}{\line(1,0){0.6}}
\put(6.2,2){\line(1,0){2}}
\put(0,2){\circle*{.3}}
 \put(2,2){\circle*{.3}}
 \put(4,2){\circle*{.3}}
 \put(6,2){\circle*{.3}}
 \put(8,2){\circle*{.3}}
 \put(10,2){\circle*{.3}}
 \put(12,2){\circle*{.3}}
 \put(14,2){\circle*{.3}}
 \put(10,4){\circle*{.3}}
 \put(12,4){\circle*{.3}}
 \put(10,2){\line(0,1){2}}
 \put(12,2){\line(0,1){2}}
 \put(8,2){\line(1,0){6}}
 \put(0,2){\makebox(1,1){$\lambda$}}
 \put(2,2){\makebox(1,1){$\lambda$}}
 \put(4,2){\makebox(1,1){$\lambda$}}
 \put(6,2){\makebox(1,1){$\lambda$}}
 \put(8,2){\makebox(1,1){$\la $}}
 \put(10,2){\makebox(1,1){$\delta$}}
 \put(12,2){\makebox(1,1){$\delta$}}
 \put(14,2){\makebox(1,1){$5$}}
 \put(10,4){\makebox(1,1){$1$}}
 \put(12,4){\makebox(1,1){$2$}}

\thinlines
\put(14,2.2){\line(-1,0){1.4}}
\put(11.8,2.2){\line(-1,0){1.4}}
\put(9.8,2.2){\line(-1,0){1.4}}
\put(9.8,2.75){\vector(-1,0){1.4}}
\put(7.8,2.2){\line(-1,0){1.4}}
\put(5.8,2.2){\line(-1,0){1.4}}
\put(3.8,2.2){\line(-1,0){1.4}}
\put(1.8,2.2){\vector(-1,0){1.4}}
\put(9.8,4){\line(0,-1){1.4}}
\put(11.8,3){\vector(-1,0){5.6}}

\put(11.8,4){\line(0, -1){1}}
\end{picture}
  \end{center}

\begin{definition}
\label{syntdebi}
We define $\La$, the {\em set of terms with de Bruijn indices},  as follows:
\begin{center}
$\La::=\mathbb{N}\:\:|\:\:(\La \La)\:\:|\:\:(\la \La)\;\;\;\;\;\;\,$ 
\end{center}
As for $\cM$, we use $ A, B,  \dots$ to range over $\La$.  We also use $m, n,  \dots$ to range 
over $\mathbb{N}$ (positive natural numbers). 
\end{definition}
We use similar notational conventions and compatibility rules as before.  However,  here we cannot compress a sequence of $\lambda$'s to one.  While we write $\la z. \la y.yz$ as $\la zy.yz$, we cannot write 
 $\la \la 1 2$  as $\la 1 2$ (which is $\la y. yx$).

\medskip

Before we can do substitution, we must learn how to update  variables.  This will be needed.  For example, $\be$-reducing 
$(\la\la {\texttt 521})(\la {\texttt 31})$ should result in  $\la {\texttt 521}$ where the variable $2$ that was bound by the $\la$ that disappeared is replaced by the argument $\la 3 1$.  But we cannot just put $\la 31$ instead of 2 in $\la 521$.  First, the 5 should be decreased by 1 because one $\la$ has disappeared from $\la \la {\texttt 521}$.  Also, the 3 of $\la 31$ should be increased by 1 when it is inserted in the hole of $\la 5[\:]1$ since the $[\:]$ is inside an extra $\la$ and so all the free variables of $\la 31$ must be increased.   The next definition introduces this updating.  
The intuition behind $\uik$ is that $k$ tests for free variables
and $i-1$ is the value by which a variable, if free, must be incremented.

\begin{definition} \label{uik}
The  {\em meta-updating functions}
$\uik :\La \rightarrow \La$ for  $k\geq 0$ and 
$i\geq 1$ are defined inductively as follows:\vspace{10pt}\\
\hspace*{25pt}$\begin{array}{l}
\uik (\ate\bte) \equiv \uik (\ate) \, \uik(\bte)\\[10pt]
\uik (\la \ate) \equiv \la (\uikl (\ate))
\end{array}$
\hspace*{16pt}$\begin{array}{l}
\uik ({\texttt n}) \equiv \left\{ \begin{array}{ll}
                           {\texttt n+i-1} & \mbox{if $\,\;n>k$}\\ 
                           {\texttt n} &\mbox{if $\,\;n\leq k\,.$}  
                         \end{array}
\right.
\end{array}$
\end{definition}

Using this updating, we define substitution in the obvious way.  The first two equalities propagate the substitution through applications
and abstractions and the last one carries out the substitution of the intended
variable (when $n=i$)
by the updated term. If the variable is not the intended one
it must be decreased by 1 if it is free (case $n>i$) because one $\la$ has 
disappeared, whereas if it is bound (case $n<i$) it must remain unaltered.

\begin{definition} \label{msubs}
The {\em meta-substitutions at level $\,i\,$}, for $\,i\geq 1\,$,  of  a term
$\,\bte\in \La\,$ in a term  $\,\ate\in\La\,$, denoted 
$\,\ate\{\!\!\{{\texttt i}\leftarrow \bte\}\!\!\}\,$\index{$\all{\ate}{i}{\bte}$}, is 
defined inductively on $A$ as follows:\\
$\begin{array}{ll}
(\tea_{1}\tea_{2})\{\!\!\{{  i}\leftarrow \bte\}\!\!\} & \equiv\:
(\tea_{1}\{\!\!\{{  i}\leftarrow \bte\}\!\!\}) \,
(\tea_{2}\{\!\!\{{  i}\leftarrow \bte\}\!\!\})\\
(\la \tea)\{\!\!\{{  i}\leftarrow \bte\}\!\!\} & \equiv\:
 \la (\tea\{\!\!\{{  i+1}\leftarrow \bte\}\!\!\})\\
{  n}\{\!\!\{{  i}\leftarrow \teb\}\!\!\}\: &\equiv \:
\left\{ \begin{array}{ll}
                           {  n-1} & \mbox{if $\,\;n>i$}\\
                           U^{i}_{0}(\teb) & \mbox{if $\,\;n=i$}\\ 
                           {  n} &\mbox{if $\,\;n<i\,.$}      
                         \end{array}
\right.
\end{array}$
\end{definition}

The following lemma establishes the properties of the meta-substit\-utions
and meta-updating functions. The proof of this lemma is  
obtained by induction on $A$ and can be found in~\cite{kamareddine/rios}
(the proof of \ref{l15} requires \ref{l14} with $p=0$;
the proof of \ref{lms} uses \ref{l13} and \ref{l15} both with $k=0$;
finally, \ref{l16} with $p=0$ is needed to prove \ref{ld}).

\begin{lemma}
\label{lemallsix}
\mbox{} \hfill
\begin{enumerate}
\item
\label{l13}
For $\,k<n\leq k+i\,$ we have:
$\uik(\tea) \equiv \all{U^{i+1}_k(\tea)}{n}{\teb}\,$.
\item
\label{l14}
For $\,p\leq k<j+p\,$ we have:
$\uik(U^j_p(\tea)) \equiv U^{j+i-1}_{p}(\tea)\,.$
\item
\label{l15}
For $\,i\leq n-k\,$ we have:
$\all{\uik(\tea)}{n}{\teb} \equiv \uik(\all{\tea}{n-i+1}{\teb})\,.$
\item
\label{lms}
[Meta-substitution lemma]
For $\,1\leq i\leq n\,$ we have:\\
$\all{\all{\tea}{i}{\teb}}{n}{\tec} \equiv \all{\all{\tea}{n+1}{\tec}}{i}{\all{\teb}{n-i+1}{\tec}}$.
\item
\label{l16}
For $\,m\leq k+1\,$ we have:
$U^i_{k+p}(U^m_p(\tea)) \equiv U^m_p(U^i_{k+p+1-m}(\tea))\,$.
\item
\label{ld}
[Distribution lemma]\\
For $\,n\leq k+1\,$ we have:
$\begin{array}{l} \uik (\all{\tea}{n}{\teb}) \equiv\\
\qquad \all{U^{i}_{k+1}(\tea)}{n}{U^i_{k-n+1} (\teb)}\,. \end{array}$
\end{enumerate}
\end{lemma}
Case~\ref{lms} is the version of Lemma~\ref{orderedmeta-sublema}.\ref{orderedmeta-sublema3} using de Bruijn indices.

\begin{definition} \label{beta}
{\em $\be_1$-reduction} is the least compatible relation on $\La$
generated by:\vspace{3pt}\\
\hspace*{35pt}($\be_1$-rule)\hspace{20pt} 
$(\la \tea)\,\teb \rightarrow_{\be_1} \all{\tea}{1}{\teb}$\vspace{3pt}\\
We define $\rrightarrow_{\be_1}$ as the reflexive transitive closure of $\rightarrow_{\be_1}$.
\end{definition}

It is easy to check ,that 
$\all{(\la {\texttt 521})}{1}{(\la {\texttt 31})} \equiv\la {\texttt 4}(\la {\texttt 41})\um$
and hence
$(\la\la {\texttt 521})(\la {\texttt 31})\rightarrow_{\be_1}\la {\texttt 4}(\la {\texttt 41})\um$.

The $\la$-calculi with variable names and with de Bruijn indices  
are isomorphic  (there are translation functions 
between $\cM$ and $\La$ which are inverses of each other and which
preserve their respective $\be$-reductions, see~\cite{mauny}). 
\subsection{The classical $\la$-calculus with de Bruijn indices and explicit substitutions}
\label{fromdeBtoexp}
 Although the $\la$-calculus with de Bruijn indices is not easy for humans, it is very straightforward for machines to implement its meta-updating, meta-substitution and beta rules.   In fact, the $\la$-calculus with de Bruijn has an important place in the implementations of functional languages. In this section, we take the classical $\la$-calculus with de Bruijn indices exactly  as it is, but simply turn its meta-updating and meta-substitution to the object level to obtain a calculus of explicit substitutions.  Extending $\la$-calculi with explicit substitutions is essential for the implementations of these calculi.
\begin{definition}[Syntax of the $\la s$-calculus]
Terms of the  $\las$-calculus are given by:\vspace{3pt}

$\Las::=\mathbb{N}\:\:|\:\:(\Las \Las)\:\:|\:\:(\la \Las)\:\:|\:\:(\Las\si\Las)\:\:|\:\:(\fik \Las)\;\;\;\;\;
\, where \; \;\;i\geq 1\, ,\;\;k\geq 0\, .$
\end{definition}
We use the notational conventions to get rid of
unnecessary parenthesis.

Now, we need to include reduction rules that operate on the new terms
built with updating and substitutions.  Definitions~\ref{uik}
and~\ref{msubs} suggest these rules.  The resulting
calculus is the explicit substitution calculus $\la s$ of~\cite{kamareddine/rios}
whose set of rules is given in Figure~\ref{figurelas}.  Note that
these rules are nothing more than $\be_1$ written now as
$\sigma$-generation, together with the rules of Definitions~\ref{uik}
and~\ref{msubs} oriented as expected.
\begin{definition}
The set of rules $\las$ is given in Figure~\ref{figurelas}.
\begin{figure}
\vspace{10pt}
\begin{center}
\begin{tabular}{lrcl}
 & & & \\[-6pt]
\hspace{8pt}{\it $\sigma$-generation} & $(\la \tea)\,\teb$ & $\lar$ & $\tea\,\sigma^1\,\teb$\\[3pt]
\hspace{8pt}{\it $\sigma$-$\la$-transition} & $(\la \tea)\si \teb$ & $\lar$ & $\la (\tea\sil \teb)$\\[3pt]
\hspace{8pt}{\it $\sigma$-app-transition} & $\;\;\;(\tea_1\,\tea_2)\si \teb$ & $\lar$ & $(\tea_1\si \teb)\,(\tea_2\si \teb)$\\[5pt]
\hspace{8pt}{\it $\sigma$-destruction} & ${\texttt n}\si \teb$ & $\lar$ & 
  $\left\{ \begin{array}{ll}
                           {\texttt n-1} & \mbox{if $\,\;n>i$}\\
                           \varphi^{i}_{0}\,\teb & \mbox{if $\,\;n=i$}\\ 
                           {\texttt n} &\mbox{if $\,\;n<i\,$}      
                         \end{array}
  \right.$\\[20pt]
\hspace{8pt}{\it $\varphi$-$\la$-transition} & $\fik(\la \tea)$ & $\lar$ & $\la (\fikl\, \tea)$\\[5pt]
\hspace{8pt}{\it $\varphi$-app-transition} & $\fik (\tea_1\,\tea_2)$ & $\lar$ & $(\fik\, \tea_1)\,(\fik\, \tea_2)$\\[5pt]
\hspace{8pt}{\it $\varphi$-destruction} & $\fik\, {\texttt n}$ & $\lar$ &
  $\left\{ \begin{array}{ll}
                           {\texttt n+i-1} & \mbox{if $\,\;n>k$}\\ 
                           {\texttt n} &\mbox{if $\,\;n\leq k\,\;\;\;$}  
                         \end{array}
 \right.$\\[-6pt]
& & & \\
\end{tabular}
\caption{The $\las$-rules}
\label{figurelas}
\end{center}
\end{figure}
The {\em $\las$-calculus} is the reduction system
$(\Las, \rightarrow_{\las})$ where $\rightarrow_{\las}$ is the least compatible
reduction on $\Las$ generated by the set of rules $\las$.
\end{definition}

\cite{kamareddine/rios} establishes that the $s$-calculus (i.e., the reduction system whose rules are those of 
Figure~\ref{figurelas} excluding $\sigma$-generation) is strongly
normalising, that the $\la s$-calculus is confluent, simulates
$\beta$-reduction and has the property of preservation of strong
normalisation PSN (i.e., if a term terminates in the calculus with de
Bruijn indices, then it terminates
in the $\la s$-calculus).

In Definition~\ref{lambdatermsdefs} we presented contexts with one hole which we used in some of our statements and proofs about $\la$-terms.  In fact, contexts (with one or more holes, also known as open terms) are an important aspect of $\la$-calculi and their implementations.  Since extending $\la$-calculi with explicit substitutions is crucial for implementations, these extensions must also cover terms with holes.  To extend the $\la s$-calculus with open terms, we need to add to the syntax of $\la s$ variables $X, Y, \cdots$ that range over terms.  However, just adding these variables to the syntax of $\la s$ and keeping the rules of Figure~\ref{figurelas} as they are does not guarantee confluence.
For example 
$((\la X)Y )\sigma^1 1 \rightarrow  (X\sigma^1 Y)\sigma^1 1$
and
$((\la X)Y )\sigma^1 1 \rightarrow  ((\la X)\sigma^1 1)(Y\sigma^1 1)$
but 
$(X\sigma^1 Y)\sigma^1 1$ and $((\la X)\sigma^1 1)(Y\sigma^1 1)$
have no common reduct.  In addition to extending the syntax with variables that range over terms, we need to extend the rules to guarantee confluence.  The extra rules needed are none other than those of Lemma~\ref{lemallsix}  oriented in the obvious way.  This results in the
calculus $\la s_e$ which is confluent on open terms~\cite{jfp}.  Like $\ls$
of~\cite{abadcardcurilevy91}, this calculus does not satisfy PSN.

\begin{definition}[The $\la s_e$-calculus]
Terms of the $\las_e$-calculus are given by:\\
$\Las_{op}::=\mbox{\bf $V$}\:\:|\:\:\mathbb{N}\:\:|\:\:(\Las_{op}\Las_{op})\:\:|\:\:(\la \Las_{op})\:\:|\:\:(\Las_{op}\sj\Las_{op})\:\:|\:\:(\fik \Las_{op})$ where $j,\,i\geq 1$, $k\geq 0$
and  {\bf V} stands for a set of variables, over which $X$, $Y$, ... range.

The set of rules $\las_e$ is obtained by adding the rules given in 
Figure~\ref{figure2} to the set $\las$ of Figure~\ref{figurelas}.
\begin{figure}
{\tiny
\begin{center}
\noindent  \begin{tabular}{lrclll}
    & & & & &\\
{\it $\sigma$-$\sigma$} & $(\tea\si \teb)\,\sigma^j\, \tec$
& 
$\lar$ &
    $(\tea\,\sigma^{j+1}\, \tec)\,\si\, (\teb\,\sigma^{j-i+1}\, \tec)$
 &if \hspace{1pt}& $i\leq j$\\
{\it $\sigma$-$\varphi_1$} & $(\fik\, \tea)\,\sigma^j\, \teb$ & $\lar$ &
    $\varphi^{i-1}_k\, \tea$ &if& $k<j< k+i$ \\
{\it $\sigma$-$\varphi_2$} & $(\fik\, \tea)\,\sigma^j\, \teb$ & $\lar$ &
    $\fik (\tea\,\sigma^{j-i+1}\, \teb)$ &if& $k+i\leq j$\\
{\it $\varphi$-$\sigma$} & $\fik (\tea\,\sigma^j\, \teb)$ & $\lar$ &
    $(\fikl\, \tea)\,\sigma^j\, (\varphi^i_{k+1-j}\, \teb)$ &if& $j\leq k+1$\\
{\it $\varphi$-$\varphi_1$} & $\fik\, (\varphi^j_l\, \tea)$ & $\lar$ &
    $\varphi^j_l\, (\varphi^i_{k+1-j}\, \tea)$ &if& $ l+j\leq k$\\
{\it $\varphi$-$\varphi_2$} & $\fik \,(\varphi^j_l\, \tea)$ & $\lar$ &
    $\varphi^{j+i-1}_l\, \tea$ &if& $l\leq k <l+j$\\ 
\end{tabular}
\caption{The new rules of the $\las_e$-calculus}
\label{figure2} 
\end{center}
}
\end{figure}
The {\em $\las_e$-calculus} is the reduction system
$(\Las_{op}, \rightarrow_{\las_e})$ where $\rightarrow_{\las_e}$ is the least compatible
reduction on $\Las_{op}$ generated by the set of rules $\las_e$.
\end{definition}

\section{Conclusion}
\label{secconc}
   In this paper we discussed a number of approaches for substitution and reduction in the $\lambda$-calculus, Most of which are influenced by the Curry school and none of which would make sense without the lessons learned from the Curry tradition.  These notes are based on students questions where they  wanted to see the build up of terms and computation stepwise from the bottom up and without hidden steps or working modulo classes. By working through grafting, replacement using ordered variable lists, and then different reduction relations based on this grafting and replacement as well as understanding the role of variable renaming,  
  the students are able to appreciate the move to manipulating terms modulo $\alpha$-classes and then seem to appreciate de Bruijn indices and the $\la$-calculus \`a la de Bruijn with or without explicit substitutions. The $\beal$ reduction we saw here is the same reduction relation given in Hindley and Seldin's book~\cite{hindley/seldin:LCCI} and is an accurate representation of how we should manage variables, substitution and alpha conversion.
 
\section*{Acknowledgements}
I am very appreciative for Jonathan Seldin for all the wonderful discussions and collaboration and friendship we have had for a quarter of a century.

\end{document}